\documentclass[11pt,a4wide]{article}

\newcommand{\remove}[1]{}
\usepackage[utf8]{inputenc}
\usepackage[english]{babel}

\usepackage{tikz,slashbox}
\usepackage{xspace}

\usepackage{amsmath,amsfonts,amssymb,epsfig,color}
\usepackage{stmaryrd}
\usepackage{amssymb}
\usepackage{amsfonts}
\usepackage{amsmath}
\usepackage{hyperref}
\usepackage{amsthm}
\usepackage{amssymb,amstext,hyperref}
\usepackage{cite}
\usepackage{graphics}
\usepackage{amsmath,dsfont}
\usepackage{mathrsfs}
\usepackage{fancyhdr}
\usepackage{verbatim}
\usepackage{boxedminipage}
\usepackage{colortbl}
\usepackage{algorithm}
\usepackage[noend]{algorithmic}

\usepackage[T1]{fontenc}
\usepackage{tikz}
\usepackage{graphicx}

\newtheorem{theorem}{Theorem}
\newtheorem{lemma}{Lemma}

\newcommand{\Ostar}{\mathcal{O}^*}

\newcommand{\Acal}{\mathcal{A}}
\newcommand{\Bcal}{\mathcal{B}}

\newcommand{\Ecal}{\mathcal{E}}
\newcommand{\Fcal}{\mathcal{F}}
\newcommand{\Gcal}{\mathcal{G}}

\newcommand{\Ical}{\mathcal{I}}

\newcommand{\Ncal}{\mathcal{N}}
\newcommand{\Ocal}{\mathcal{O}}
\newcommand{\Pcal}{\mathcal{P}}
\newcommand{\Qcal}{\mathcal{Q}}
\newcommand{\Rcal}{\mathcal{R}}

\newcommand{\Xcal}{\mathcal{X}}

\newcommand{\Nbb}{\mathbb{N}}

\newcommand{\sbf}{{\bf s}}
\newcommand{\rbf}{{\bf r}}

\newcommand{\dptable}{{\bf r}}

\newcommand{\ETH}{{\sf ETH}\xspace}

\newcommand{\degs}[2]{{\sf deg}_{#1}(#2)}

\theoremstyle{plain}
\newtheorem{proposition}[theorem]{Proposition}
\newtheorem{observation}{Observation}

\newcommand{\paraprobl}[5]
{
  \begin{flushleft}
    \fbox{
      \begin{minipage}{#5cm}
        \noindent {\textsc {#1}}\\
        {\bf Input:} #2\\
        {\bf Parameter:} #4\\
        {\bf Output:} #3
      \end{minipage}
    }
  \end{flushleft}
}

\newcommand{\sm}{\setminus}
\newcommand{\gm}{\setminus}

\newcommand{\tw}{{\sf{tw}}}

\newcommand{\es}{\emptyset}
\newcommand{\ef}{\varnothing}

\newcommand{\intv}[1]{\left [ #1 \right ]}

\renewcommand{\angle}[1]{\!\left \langle #1 \right \rangle}

\newcommand{\pretp}{\preceq_{\sf tm}}
\newcommand{\prem}{\preceq_{\sf m}}

\newcommand{\paw}{{\sf paw}\xspace}

\newcommand{\chair}{{\sf chair}\xspace}

\newcommand{\banner}{{\sf banner}\xspace}

\newcommand{\dart}{{\sf dart}\xspace}

\newcommand{\gem}{{\sf gem}\xspace}
\newcommand{\house}{{\sf px}\house}
\newcommand{\ourdiamond}{{\sf diamond}\xspace}

\usepackage{aecompl}
\usepackage{color}
\definecolor{linkcol}{rgb}{0,0,0.8}
\definecolor{citecol}{rgb}{0.65,0,0}
\definecolor{titlecol}{rgb}{0.65,0,0}

\hypersetup {bookmarksopen=true, pdftitle="Hitting (topological) minors on
  bounded treewidth graphs", pdfauthor="Julien Baste - Ignasi Sau - Dimitrios M. Thilikos",
  pdfsubject="Hitting (topological) minors on
  bounded treewidth graphs", 
  pdfmenubar=true, 
  pdfhighlight=/O, 
  colorlinks=true, 
  pdfpagemode=None, 
  pdfpagelayout=DoublePage, 
  pdffitwindow=true, 
  linkcolor=linkcol, 
  citecolor=citecol, 
  urlcolor=linkcol 
}

\newcommand{\vertex}[1]{\filldraw (#1) circle (2 pt);}

\newcommand{\bct}{\textsf{bct}}

\newcommand{\ex}{\textsf{ex}}

\newcommand{\rmc}{\textsf{rmc}}
\newcommand{\ins}{\textsf{ins}}
\newcommand{\shift}{\textsf{shft}}
\newcommand{\glue}{\textsf{glue}}
\newcommand{\proj}{\textsf{proj}}
\newcommand{\join}{\textsf{join}}
\newcommand{\opt}{\textsf{opt}}
\newcommand{\reduce}{\textsf{reduce}}

\newcommand{\hs}{\widehat{S}}

\newcommand{\cc}{{\sf cc}}
\newcommand{\ct}{{\sf c}_3}

\definecolor{gray0}{gray}{0.875}
\definecolor{gray1}{gray}{0.775}
\definecolor{gray2}{gray}{0.75}

\newcommand\cuparrow{%
  \mathrel{\ooalign{\hss$\cup$\hss\cr%
      \kern0.3ex\raise0.7ex\hbox{\scalebox{0.7}{$\downarrow$}}}}}

\newcommand\bigcuparrow{%
  \mathrel{\ooalign{\hss$\bigcup$\hss\cr%
      \kern0.55ex\raise0.7ex\hbox{\scalebox{0.7}{$\downarrow$}}}}}

\newcommand{\pbtm}{\textsc{$\Fcal$-TM-Deletion} }
\newcommand{\pbm}{\textsc{$\Fcal$-M-Deletion} }

\begin{document}

\title{\vspace{-.5cm}Hitting minors on bounded treewidth graphs.\\II. Single-exponential algorithms\thanks{Emails of authors: \texttt{julien.baste@uni-ulm.de}, \texttt{ignasi.sau@lirmm.fr}, \texttt{sedthilk@thilikos.info}.\vspace{.2cm}\newline \indent \!\!\!\!\!The results of this article are permanently available at \texttt{https://arxiv.org/abs/1704.07284}. Extended abstracts containing some of the results of this article appeared in the \emph{Proc. of the 12th International Symposium on Parameterized and Exact Computation  (\textbf{IPEC 2017})}~\cite{BasteST17} and in the \emph{Proc. of the 13th International Symposium on Parameterized and Exact Computation  (\textbf{IPEC 2018})}~\cite{BasteST18}. Work supported by French projects DEMOGRAPH (ANR-16-CE40-0028) and ESIGMA (ANR-17-CE23-0010).\newline}}

\author{\bigskip Julien Baste\thanks{LIRMM, Université de Montpellier, Montpellier, France.}$\ ^{,}$\thanks{Sorbonne Université, Laboratoire d'Informatique de Paris 6, LIP6, Paris, France.} \and
  Ignasi Sau\thanks{LIRMM,  Université de Montpellier, CNRS, Montpellier, France.}
  \and
  Dimitrios  M. Thilikos$^{{\small \S}}$}

\date{\vspace{-1cm}}

\maketitle

\begin{abstract}
   \noindent For a finite collection of graphs ${\cal F}$, the \textsc{$\Fcal$-M-Deletion} (resp. \textsc{$\Fcal$-TM-Deletion}) problem consists in, given a graph $G$ and an integer $k$, decide  whether there exists $S \subseteq V(G)$ with $|S| \leq k$ such that $G \setminus S$ does not contain any of the graphs in ${\cal F}$ as a minor (resp. topological minor). We are interested in the parameterized complexity of both problems when the parameter is the treewidth of $G$, denoted by $\tw$, and specifically in
   the cases where $\Fcal$  contains a single connected planar graph $H$. We present algorithms  running in time $2^{\Ocal(\tw)} \cdot n^{\Ocal(1)}$, called \emph{single-exponential}, when $H$ is either $P_3$, $P_4$, $C_4$, the \paw, the \chair, and the \banner for both \textsc{$\{H\}$-M-Deletion} and \textsc{$\{H\}$-TM-Deletion}, and when $H=K_{1,i}$, with $i \geq 1$, for \textsc{$\{H\}$-TM-Deletion}. Some of these algorithms use the rank-based approach introduced by Bodlaender et al.~[Inform Comput, 2015].
 This is the second of a series of articles on this topic, and the results given here together with other ones allow us, in particular, to provide a tight dichotomy on the complexity of \textsc{$\{H\}$-M-Deletion} in terms of $H$.


\vspace{.5cm}

\vspace{.5cm}

\end{abstract}

\newpage

\section{Introduction}
\label{illusion}

Let ${\cal F}$ be a finite non-empty collection of non-empty graphs.  In the \textsc{$\Fcal$-M-Deletion} (resp. \textsc{$\Fcal$-TM-Deletion}) problem, we are given a graph $G$ and an integer $k$, and the objective is to decide whether there exists a set $S \subseteq V(G)$ with $|S| \leq k$ such that $G \setminus S$ does not contain any of the graphs in ${\cal F}$ as a minor (resp. topological minor). These problems have a big expressive power, as instantiations of them correspond to several well-studied problems. For instance,  the cases ${\cal F}= \{K_2\}$, ${\cal F}= \{K_3\}$, and ${\cal F}= \{K_5,K_{3,3}\}$ of \textsc{$\Fcal$-M-Deletion} (or \textsc{$\Fcal$-TM-Deletion}) correspond to \textsc{Vertex Cover}, \textsc{Feedback Vertex Set}, and \textsc{Vertex Planarization}, respectively.
For the sake of readability, we use the notation \textsc{$\Fcal$-Deletion} in statements that apply to {\sl both} \textsc{$\Fcal$-M-Deletion} and  \textsc{$\Fcal$-TM-Deletion}.



We are interested in the parameterized complexity of \textsc{$\Fcal$-Deletion} when the parameter is the treewidth of the input graph.  Courcelle's theorem~\cite{Courcelle90} implies that \textsc{$\Fcal$-Deletion} can be solved in time $\Ostar(f(\tw))$ on graphs with treewidth at most $\tw$, where $f$ is some computable function\footnote{The notation $\Ostar(\cdot)$  suppresses polynomial factors depending on the size of the input graph.}. Our objective is to determine, for a fixed collection ${\cal F}$, which is the {\sl smallest} such function $f$ that one can (asymptotically) hope for, subject to reasonable complexity assumptions.

This line of research has recently attracted some attention  in the parameterized complexity community. For instance, \textsc{Vertex Cover} is easily solvable in time $\Ostar(2^{\Ocal(\tw)})$, called \emph{single-exponential}, by standard dynamic-programming techniques, and no algorithm with running time $\Ostar(2^{o(\tw)})$ exists, unless the Exponential Time Hypothesis (\ETH)\footnote{The \ETH states that 3-\textsc{SAT} on $n$ variables cannot be solved in time $2^{o(n)}$; see~\cite{ImpagliazzoP01} for more details.} fails~\cite{ImpagliazzoP01}.
For \textsc{Feedback Vertex Set}, standard dynamic programming techniques give a running time of $\Ostar(2^{\Ocal(\tw \cdot \log \tw)})$, while the lower bound under the \ETH~\cite{ImpagliazzoP01} is again $\Ostar(2^{o(\tw)})$. This gap remained open for a while, until Cygan et al.~\cite{CyganNPPRW11} presented an optimal algorithm running in time $\Ostar(2^{\Ocal(\tw)})$, introducing  the celebrated \emph{Cut{\sl \&}Count}
technique, which produces randomized algorithms. This article triggered several other techniques to obtain single-exponential {\sl deterministic} algorithms for so-called \emph{connectivity problems} on graphs of bounded treewidth, mostly based on algebraic tools~\cite{BodlaenderCKN15,FominLPS16}. We refer the reader to~\cite{monster1} for a more detailed discussion about related work. In particular, in this article we make use of one of the techniques presented by Bodlaender et al.~\cite{BodlaenderCKN15}, called \emph{rank-based approach}. It is worth mentioning that this approach has been recently applied to dense graph classes, namely those with structured neighborhoods~\cite{BeKa18}.



\medskip
\noindent
\textbf{Our results and techniques}. We provide several single-exponential algorithms when $\Fcal$ contains a single connected planar graph $H$. Namely, we show that if $\Fcal \in \{\{P_3\}, \{P_4\}, \{K_{1,i}\},$ $\{C_4\}\, \{\paw\}, \{\chair\}, \{\banner\}\}$ (see Figure~\ref{shifting} for an illustration of these graphs), then \pbtm 
can be solved in single-exponential time. Note that all these graphs have maximum degree at most three, except $K_{1,i}$ for $i \geq 4$, and therefore the corresponding algorithms also apply to the \pbm problem. Indeed, for graphs $H$ with maximum degree at most three, containing $H$ as a minor is equivalent to containing $H$ as a topological minor.
The fact that we are not able to provide  single-exponential algorithms for  \textsc{$\{K_{1,i}\}$-M-Deletion} with $i \geq 4$ seems to be unavoidable: we prove in~\cite{monster3} that there is no algorithm in time $\Ostar(2^{o(\tw \cdot \log \tw )})$ for these cases, unless the \ETH fails. This exhibits, to the best of our knowledge, the first difference between the computational complexity of both problems.

\medskip

The single-exponential algorithms presented in this article are ad hoc, some being easier than others. All of them exploit a structural characterization of the graphs that exclude that particular graph $H$ as a (topological) minor; cf. for instance Lemmas~\ref{headlong} and~\ref{voltaire}. Intuitively, the ``complexity'' of this characterization is what determines the difficulty of the corresponding dynamic programming algorithm, and is also what makes the difference between being solvable in single-exponential time or not.

More precisely, the algorithms for $\{P_3\}$-\textsc{Deletion}, $\{P_4\}$-\textsc{Deletion}, and  $\{K_{1,i}\}$-\textsc{TM-Deletion} use standard (but non-trivial) dynamic programming techniques on graphs of bounded treewidth, exploiting the simple structure of graphs that do not contain these particular graphs as a topological minor (or as a subgraph, which in these cases is equivalent). The algorithms for $\{P_3\}$-\textsc{Deletion} and $\{K_{1,i}\}$-\textsc{TM-Deletion} are quite simple, while the one for $\{P_4\}$-\textsc{Deletion} is slightly more technical.

The algorithms for $\{C_4\}$-\textsc{Deletion} and $\{\paw\}$-\textsc{Deletion} are  more involved, and use the rank-based approach introduced by Bodlaender et al.~\cite{BodlaenderCKN15}, exploiting again the structure of graphs that do not contain $C_4$ or the $\paw$ as a minor (cf. Lemma~\ref{lifeboat} and~\ref{absurdly}, respectively). It might seem counterintuitive that this technique works for $C_4$, and stops working for $C_i$ with $i \geq 5$. A possible reason for that is that the only cycles of a $C_4$-minor-free graph are triangles and each triangle must be  contained in a bag of a tree decomposition. This property, which is not true anymore for $C_i$-minor-free graphs with $i \geq 5$, permits to keep track of the structure of partial solutions with tables of small size. The algorithm for $\{\paw\}$-\textsc{Deletion}  combines classical dynamic programming techniques  and the rank-based approach.

Finally, the algorithms for $\{\chair\}$-\textsc{Deletion} and $\{\banner\}$-\textsc{Deletion} are a combination of the above ones, the latter one using again the rank-based approach. Given the large amount of labels that we need in the tables and the similarity with other algorithms for which we provide all the details, we only present a sketch of these two algorithms.

\newpage
\noindent
\textbf{Results in other articles of the series and discussion}. In the first article of this series~\cite{monster1}, we show, among other results, that for every collection
$\mathcal{F}$ containing at least one planar graph (resp. subcubic planar graph), \textsc{$\Fcal$-M-Deletion} (resp. \textsc{$\Fcal$-TM-Deletion}) can be solved in time
  $\Ostar(2^{\Ocal(\tw  \cdot \log \tw)})$. In the third article of this series~\cite{monster3}, we focus on lower bounds under the \ETH. Namely, we prove that for any connected\footnote{A \emph{connected} collection $\mathcal{F}$ is a collection containing only connected graphs.} $\Fcal$, \textsc{$\Fcal$-Deletion} cannot be solved in time $\Ostar(2^{o(\tw)})$, even if the input graph $G$ is planar, and we provide superexponential lower bounds for a number of collections $\Fcal$. In particular, we prove a lower bound of $\Ostar(2^{o(\tw \cdot \log \tw)})$ when $\Fcal$ contains a single connected graph that is either $P_5$ or is not a minor of the $\banner$, with the exception of $K_{1,i}$ for the topological minor version. These lower bounds, together with the ad hoc single-exponential  algorithms given in this article and the general algorithms described in~\cite{monster1}, cover all the cases of \textsc{$\Fcal$-M-Deletion} where $\Fcal$ contains a single connected planar graph $H$, yielding a dichotomy in terms of  $H$. In the fourth article of this series~\cite{BasteST20-SODA} (whose full version is~\cite{SODA-arXiv}), we presented an algorithm for \textsc{$\Fcal$-M-Deletion} in time $\Ostar(2^{O(\tw \cdot \log \tw)})$ for {\sl any} collection $\Fcal$, yielding together with the lower bounds in~\cite{monster3} and the results of the current article a dichotomy for \textsc{$\Fcal$-M-Deletion} where $\Fcal$ consists of a single connected (non-necessarily planar) graph $H$. 
  Namely, as stated in~\cite{SODA-arXiv}, if $H$ is a connected graph on at least two vertices, then the \textsc{$\{H\}$-M-Deletion} problem  can be solved in time
    \begin{itemize}
    \item $\Ostar(2^{\Theta(\tw )})$, if $H$ is a contraction of the \chair or the \banner, and
    \item $\Ostar(2^{\Theta(\tw \cdot \log \tw)})$, otherwise.
    \end{itemize}

    In the above statements, we use the $\Theta$-notation to indicate that these algorithms are {\sl optimal} under the \ETH. Note that the first item is equivalent to $H$ being a minor of the \banner that is different from $P_5$. This dichotomy is depicted in Figure~\ref{shifting}, containing all  connected  graphs $H$ with $2 \leq |V(H)| \leq 5$; note that if $|V(H)| \geq 6$, then $H$ is not a contraction of the \chair or the \banner, and therefore the second item above applies. Note also that $K_4$ and the {\sf diamond} are the only graphs on at most four vertices for which the problem is solvable in time $\Ostar(2^{\Theta (\tw \cdot \log \tw)})$ and that the \chair and the \banner are the only graphs on at least  five vertices for which the problem is solvable in time $\Ostar(2^{\Theta (\tw)})$. Note also that the cases $\Fcal = \{P_2\}$~\cite{ImpagliazzoP01,CyganFKLMPPS15}, $\Fcal = \{P_3\}$~\cite{P3-cover,P3-cover-improved}, and $\Fcal = \{C_3\}$~\cite{CyganNPPRW11,BodlaenderCKN15} were already known.

    The crucial role played by the \banner and $P_5$ (or equivalently, the \chair and the \banner) in the complexity dichotomy may seem surprising at first sight. In fact, we realized a posteriori that the ``easy'' cases can be succinctly described in terms of the \banner and $P_5$ by taking a look at Figure~\ref{shifting}. Nevertheless, there is some intuitive reason for which excluding the \banner constitutes the horizon on the existence of single-exponential algorithms (forgetting about the ``exception'' $\Fcal = \{P_5\}$). Namely, every connected component of a graph that excludes the \banner as a (topological) minor is either a cycle (of any length) or a tree in which some vertices have been replaced by triangles; both such types of components can be maintained by a dynamic programming algorithm in single-exponential time. It appears that if the characterization of the allowed connected components is enriched in some way, such as restricting the length of the allowed cycles or forbidding certain degrees, the problem becomes inherently more difficult.

    \begin{figure}[h!]
  \begin{center}
    \includegraphics[width=.87\textwidth]{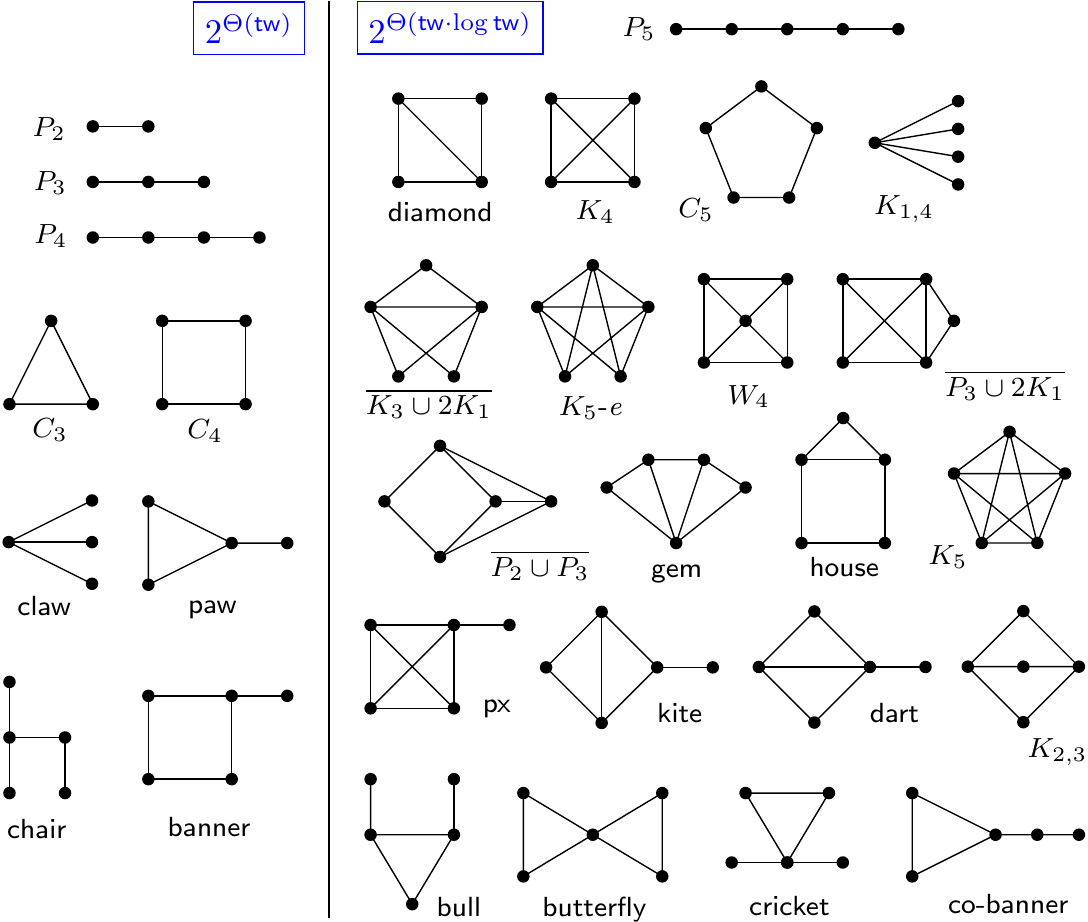}
  \end{center}\vspace{-.25cm}
  \caption{Classification of the complexity of \textsc{$\{H\}$-M-Deletion} for all connected simple graphs $H$ with $2 \leq |V(H)|\leq 5$: for the nine graphs on the left (resp. 21 graphs on the right, and all the larger ones), the problem is solvable in time $2^{\Theta(\tw)} \cdot n^{\Ocal(1)}$ (resp. $2^{\Theta(\tw \cdot \log \tw)} \cdot n^{\Ocal(1)}$). For \textsc{$\{H\}$-TM-Deletion}, $K_{1,4}$ should be on the left. This figure also appears in~\cite{monster3}.}
  \label{shifting}
  \vspace{.35cm}
\end{figure}

\newpage 

\noindent
\textbf{Organization of the paper}. In Section~\ref{banality} we give some preliminaries. We deal with $P_3$, $P_4$, $K_{1,s}$, $C_4$, the \paw, the \chair, and the \banner in Sections~\ref{accuracy}, \ref{affected}, \ref{repaired}, \ref{benjamin}, \ref{nugatory}, \ref{provider}, and~\ref{sec:banner}, respectively.
We conclude  in Section~\ref{upheaval} with some  questions for further research.

\section{Preliminaries}
\label{banality}




In this section we provide some preliminaries to be used in the following sections.


\bigskip
\noindent
\textbf{Sets, integers, and functions.}
We denote by $\Nbb$ the set of every non-negative integer and we
set $\Nbb^+=\Nbb\setminus\{0\}$.
Given two integers $p$ and $q$, the set $\intv{p,q}$
refers to the set of every integer $r$ such that $p \leq r \leq q$.
Moreover, for each integer $p \geq 1$, we set $\Nbb_{\geq p} = \Bbb{N}\setminus\intv{0,p-1}$.


We use $\emptyset$ to denote the empty set and
$\varnothing$ to denote the empty function, i.e., the unique subset of $\emptyset\times\emptyset$.
Given a function $f:A\to B$ and a set $S$, we define $f|_{S}=\{(x,f(x))\mid x\in S\cap A\}$.
Moreover if $S \subseteq A$, we set $f(S) = \bigcup_{s \in S} \{f(s)\}$.
Given a set $S$, we denote by ${S \choose 2}$ the set containing
every subset of $S$ that has cardinality two.



\bigskip
\noindent
\textbf{Graphs}.  All the graphs that we consider in this paper are undirected, finite, and without  loops or multiple edges.
 We use standard graph-theoretic notation, and we refer the reader to~\cite{Die10} for any undefined terminology. Given a graph $G$, we denote by $V(G)$ the set of vertices of  $G$ and by $E(G)$ the set of the edges of $G$.
We call $|V(G)|$ {\em the size} of $G$. A graph is {\em the empty} graph if its size is zero.
We also denote by $L(G)$ the set of the vertices of $G$ that have degree exactly
ones.
If $G$ is a tree (i.e., a connected acyclic graph) then $L(G)$ is the set of the {\em leaves}
of $G$.
A {\em vertex labeling} of $G$ is some injection $\rho: V(G)\to\Bbb{N}^{+}$. 
{Given a vertex $v \in V(G)$, we define the \emph{neighborhood} of
  $v$ as $N_G(v) = \{u \mid u \in V(G), \{u,v\} \in E(G)\}$ and the \emph{closed neighborhood}
  of $v$ as $N_G[v] = N_G(v) \cup \{v\}$.} If $X\subseteq V(G)$, then we write $N_{G}(X)=(\bigcup_{v\in X}N_{G}(v))\setminus X$.
The {\em degree} of a vertex $v$ in $G$ is defined as $\degs{G}{v}=|N_{G}(v)|$. A graph is called {\em subcubic}
if all its vertices have degree at most three.

A \emph{subgraph} $H = (V_H,E_H)$ of a graph $G=(V,E)$ is a graph such that $V_H \subseteq V(G)$ and $E_H \subseteq E(G) \cap {V(H) \choose 2}$.
If $S \subseteq V(G)$, the subgraph of $G$ \emph{induced by} $S$, denoted $G[S]$, is the graph $(S, E(G) \cap {S \choose 2})$.
We also define $G \gm S$ to be the subgraph of $G$ induced by $V(G) \sm S$.
If $S \subseteq E(G)$, we denote by $G \gm S$ the graph $(V(G), E(G) \sm S)$.
%

If $s,t\in V(G)$, an {\em $(s,t)$-path} of $G$
is any connected subgraph $P$ of $G$ with maximum degree two and where $s,t\in L(P)$.
We say that two vertices $s$ and $t$ are \emph{connected} in $G$ if $G$ contains an $(s,t)$-path as a subgraph.
We finally denote by ${\cal P}(G)$ the set of all paths of $G$.
{Given $P \in \Pcal(G)$, we say that $v \in V(P)$ is an \emph{internal vertex} of $P$ if $\degs{P}{v} = 2$.}
{Given an integer $i$ and a graph $G$, we say that $G$ is $i$-connected if for each $\{u,v\} \in {V(G) \choose 2}$,
  there exists a set $\Qcal \subseteq \Pcal(G)$ of $(u,v)$-paths of $G$ such that
  $|\Qcal| = i$ and
  for each $P_1,P_2 \in \Qcal$ such that $P_1 \not = P_2$, $V(P_1) \cap V(P_2) = \{u,v\}$.}
We denote by $K_{r}$,  $P_{r}$, and $C_{r}$, the complete graph, the path, and the cycle  on $r$ vertices, respectively.

\bigskip
\noindent
\textbf{Minors and topological minors.}
Given two graphs $H$ and $G$ and two functions $\phi: V(H)\to V(G)$ and $\sigma: E(H)\to{\cal P}(G)$, we say that $(\phi,\sigma)$ is {\em a topological minor model of $H$ in $G$}
if
\begin{itemize}
\item for every $\{x,y\}\in E(H),$ $\sigma(\{x,y\})$ is an $(\phi(x),\phi(y))$-path 
  in $G$ and
\item if $P_{1},P_{2}$ are two distinct paths in $\sigma(E(H))$, then none of the internal vertices of $P_{1}$
  is a vertex of $P_{2}$.
\end{itemize}

The {\em branch} vertices of
$(\phi,\sigma)$ are the vertices in  $\phi(V(E))$, while the {\em subdivision} vertices of $(\phi,\sigma)$ are the internal vertices of the paths in $\sigma(E(H))$.

We say that $G$ contains $H$ as a \emph{topological minor}, denoted by $H\preceq_{\sf tm} G$, if there
is a topological minor model  $(\phi,\sigma)$ of $H$ in $G$.


Given two graphs $H$ and $G$ and a function $\phi: V(H)\to 2^{V(G)}$, we say that $\phi$ is {\em a minor model of $H$ in $G$}
if
\begin{itemize}
\item for every $x \in V(H)$, $G[\phi(x)]$ is a connected non-empty graph and
\item for every $\{x,y\} \in E(H)$, there exist  $x' \in \phi(x)$ and $y' \in \phi(y)$ such that $\{x',y'\} \in E(G)$.
\end{itemize}

We say that $G$ contains $H$ as a \emph{minor}, denoted by $H\prem G$, if there
is a minor model  $\phi$ of $H$ in $G$.

\bigskip
\noindent
\textbf{Graph collections.}
Let ${\cal F}$  be a collection of graphs. From now on instead of ``collection of graphs''  we use the shortcut ``collection''.
If ${\cal F}$ is a collection that is finite, non-empty, and all its graphs are  non-empty, then we say that ${\cal F}$
is a {\em proper collection}. For any proper collection ${\cal F}$, we  define
${\sf size}({\cal F})=\max\{\{|V(H)|\mid H\in \cal F\}\cup\{|{\cal F}|\}\}$.
Note that if the size of ${\cal F}$ is bounded, then
the size of the graphs in ${\cal F}$ is also bounded.
We say that ${\cal F}$ is a {\em  planar collection} (resp. {\em planar subcubic collection}) if it is proper and
{\sl at least one} of the graphs in ${\cal F}$ is  planar (resp. planar and subcubic).
We say that ${\cal F}$ is a {\em connected collection} if it is proper and
{\sl all} the graphs in ${\cal F}$ are connected.
We  say that $\Fcal$ is an {\em (topological) minor antichain}
if no two of its elements are comparable via the  (topological)  minor relation.

Let $\Fcal$ be a proper collection. We extend the (topological) minor relation to $\Fcal$ such that, given a graph $G$,
$\Fcal \preceq_{\sf tm} G$ (resp. $\Fcal \preceq_{\sf m} G$) if and only if there exists a graph $H \in \Fcal$ such that $H \preceq_{\sf tm} G$ (resp. $H \preceq_{\sf m} G$).
We also denote $\ex_{\sf tm}(\Fcal)=\{G\mid \Fcal\npreceq_{\sf tm} G \}$, i.e.,
$\ex_{\sf tm}(\Fcal)$ is the class of graphs that do not contain any graph in $\Fcal$ as a topological minor.
The set $\ex_{\sf m}(\Fcal)$ is defined analogously.

\bigskip
\noindent\textbf{Tree decompositions.} A \emph{tree decomposition} of a graph $G$ is a pair ${\cal D}=(T,{\cal X})$, where $T$ is a tree
and ${\cal X}=\{X_{t}\mid t\in V(T)\}$ is a collection of subsets of $V(G)$
such that:
\begin{itemize}
\item $\bigcup_{t \in V(T)} X_t = V(G)$,
\item for every edge $\{u,v\} \in E$, there is a $t \in V(T)$ such that $\{u, v\} \subseteq X_t$, and
\item for each $\{x,y,z\} \subseteq V(T)$ such that $z$ lies on the unique path between $x$ and $y$ in $T$,  $X_x \cap X_y \subseteq X_z$.
\end{itemize}
We call the vertices of $T$ {\em nodes} of ${\cal D}$ and the sets in ${\cal X}$ {\em bags} of ${\cal D}$. The \emph{width} of a  tree decomposition ${\cal D}=(T,{\cal X})$ is $\max_{t \in V(T)} |X_t| - 1$.
The \emph{treewidth} of a graph $G$, denoted by $\tw(G)$, is the smallest integer $w$ such that there exists a tree decomposition of $G$ of width at most $w$.
For each $t \in V(T)$, we denote by $E_t$ the set $E(G[X_t])$.

We need to introduce nice tree decompositions, which will make the presentation of the algorithms much simpler.

\medskip

\noindent
\textbf{Nice tree decompositions.} Let ${\cal D}=(T,{\cal X})$
be a tree decomposition of $G$, $r$ be a vertex of $T$, and   ${\cal G}=\{G_{t}\mid t\in V(T)\}$ be
a collection of subgraphs of   $G$, indexed by the vertices of $T$.
We say that the triple $({\cal D},r,{\cal G})$ is a
\emph{nice tree decomposition} of $G$ if the following conditions hold:
\begin{itemize}

\item $X_{r} = \es$ and $G_{r}=G$,
\item each node of ${\cal D}$ has at most two children in $T$,
\item for each leaf $t \in V(T)$, $X_t = \es$ and $G_{t}=(\emptyset,\emptyset).$ Such $t$ is called a {\em leaf node},
\item if $t \in V(T)$ has exactly one child $t'$, then either
  \begin{itemize}
  \item $X_t = X_{t'}\cup \{v_{\rm insert}\}$ for some $v_{\rm insert} \not \in X_{t'}$ and $G_{t}=G[V(G_{t'})\cup\{v_{\rm insert}\}]$.
    The node $t$ is called \emph{introduce vertex}  node   and the vertex $v_{\rm insert}$ is the {\em insertion vertex} of $X_{t}$,
  \item $X_t = X_{t'} \sm \{v_{\rm forget}\}$ for some $v_{\rm forget} \in X_{t'}$ and $G_{t}=G_{t'}$.
    The node $t$ is called   \emph{forget vertex} node and $v_{\rm forget}$ is the {\em forget vertex} of $X_{t}$.
  \end{itemize}
\item if $t \in V(T)$ has exactly two children $t'$ and $t''$, then $X_{t} = X_{t'} = X_{t''}$, $E(G_{t'})\cap E(G_{t''})=E(G[X_t])$, and $G_t = (V(G_{t'})\cup V(G_{t''}),E(G_{t'})\cup E(G_{t''}))$. The node $t$ is called a \emph{join} node.
\end{itemize}



For each $t \in V(T)$, we denote by $V_t$ the set $V(G_t)$.
As discussed in~\cite{Klo94}, given a tree decomposition, it is possible to transform it in polynomial time to a {\sl nice} new one of the same width. Moreover, by Bodlaender et al.~\cite{BodlaenderDDFLP16} we can find in time $2^{\Ocal(\tw)}\cdot n$ a tree decomposition of width $\Ocal(\tw)$ of any graph $G$. Hence, since in this section we focus on single-exponential algorithms, we may assume that a nice tree decomposition of width $w = \Ocal(\tw)$
is given with the input.




\medskip

We also need the following simple observation that will be implicitly used  in the algorithms of Sections~\ref{accuracy},~\ref{affected},~and~\ref{benjamin}.



\begin{observation}
  Let $G$ be a graph and $h$ be a positive integer.
  Then the following assertions are equivalent.
  \begin{itemize}
  \item $G$ contains $P_h$ as a topological minor.
  \item $G$ contains $P_h$ as a minor.
  \item $G$ contains $P_h$ as a subgraph.
  \end{itemize}
  Moreover,  the following assertions are also equivalent.
  \begin{itemize}
  \item $G$ contains $C_h$ as a topological minor.
  \item $G$ contains $C_h$ as a minor.
  \end{itemize}
\end{observation}


\bigskip
\noindent\textbf{Parameterized complexity.} We refer the reader to~\cite{DF13,CyganFKLMPPS15} for basic background on parameterized complexity, and we recall here only some very basic definitions.
A \emph{parameterized problem} is a language $L \subseteq \Sigma^* \times \mathbb{N}$.  For an instance $I=(x,k) \in \Sigma^* \times \mathbb{N}$, $k$ is called the \emph{parameter}. 
A parameterized problem is \emph{fixed-parameter tractable} ({\sf FPT}) if there exists an algorithm $\Acal$, a computable function $f$, and a constant $c$ such that given an instance $I=(x,k)$,
$\Acal$ (called an {\sf FPT} \emph{algorithm}) correctly decides whether $I \in L$ in time bounded by $f(k) \cdot |I|^c$.


\bigskip
\noindent\textbf{Main ingredients of the rank-based approach.} We are now going to restate the tools introduced by Bodlaender et al.~\cite{BodlaenderCKN15} that we need for our purposes.

Let $U$ be a set.
We define $\Pi(U)$ to be the set of all partitions of $U$.
Given two partitions $p$ and $q$ of $U$, we define the coarsening relation $\sqsubseteq$ such that $p \sqsubseteq q$ if for each $S \in q$, there exists $S' \in p$ such that $S \subseteq S'$.
$(\Pi(U),\sqsubseteq)$ defines a lattice with minimum element $\{\{U\}\}$ and maximum element $\{\{x\}\mid x \in U\}$.
On this lattice, we denote by $\sqcap$ the meet operation and by $\sqcup$ the join operation.

Let $p \in \Pi(U)$.
For $X \subseteq U$ we denote by $p_{\downarrow X} = \{S \cap X \mid S \in p, S \cap X \not = \es \}\in \Pi(X)$
the partition obtained by removing all elements not in $X$ from $p$, and analogously for
$U \subseteq X$ we denote $p_{\uparrow X} = p \cup \{\{x\}\mid x \in X \sm U\}\in \Pi(X)$  the partition obtained by adding to $p$ a singleton for each element in
$ X \sm U$.
Given a subset $S$ of $U$, we define the partition $U[S] = \{\{x\}\mid x \in U\sm S\} \cup \{S\}$.

A set of \emph{weighted partitions} is a set $\Acal \subseteq \Pi(U) \times \Nbb$.
We also define $\rmc(\Acal) = \{(p,w) \in \Acal \mid \forall (p',w') \in \Acal: p'=p \Rightarrow w \leq w'\}$.

We now define some operations on weighted partitions.
Let $U$ be a set and $\Acal \subseteq \Pi(U) \times \Nbb$.

\begin{description}
\item[Union.] Given $\Bcal \subseteq \Pi(U) \times \Nbb$, we define $\Acal \cuparrow \Bcal = \rmc(\Acal \cup \Bcal)$.
\item[Insert.] Given a set $X$ such that $X \cap U = \es$, we define $\ins(X,\Acal)= \{(p_{\uparrow U \cup X},w) \mid (p,w) \in \Acal\}$.
\item[Shift.] Given $w' \in \Nbb$, we define $\shift(w',\Acal) = \{(p,w+w') \mid (p,w) \in \Acal\}$.
\item[Glue.] Given a set $S$, we define $\hat{U} = U \cup S$ and $\glue(S,\Acal) \subseteq \Pi(\hat{U}) \times \Nbb$ as\\
  $\glue(S, \Acal) = \rmc(\{(\hat{U}[S] \sqcap p_{\uparrow \hat{U}}, w\mid (p,w) \in \Acal\})$.\\
  Given $w: \hat{U} \times \hat{U} \to \Ncal$, we define $\glue_w(\{u,v\},\Acal) = \shift(w(u,v), \glue(\{u,v\},\Acal))$.
\item[Project.] Given $X \subseteq U$, we define $\overline{X} = U \sm X$ and $\proj(X, \Acal) \subseteq \Pi(\overline{X}) \times \Nbb$ as\\
  $\proj(X,\Acal) = \rmc(\{(p_{\downarrow \overline{X}},w) \mid (p,w) \in \Acal, \forall e \in X : \forall e' \in \overline{X} : p \sqsubseteq U[ee']\})$.
\item[Join.] Given a set $U'$, $\Bcal \subseteq \Pi(U') \times \Nbb$, and $\hat{U} = U \cup U'$, we define $\join(\Acal, \Bcal) \subseteq \Pi(\hat{U}) \times \Nbb$ as\\
  $\join(\Acal, \Bcal) = \rmc(\{(p_{\uparrow \hat{U}}\sqcap q_{\uparrow \hat{U}}, w_1+w_2) \mid (p,w_1) \in \Acal, (q,w_2) \in \Bcal\})$.
\end{description}


\begin{proposition}[Bodlaender et al.~\cite{BodlaenderCKN15}]
  \label{broadest}
  Each of the operations union, insert, shift, glue, and project can be carried out in time $s\cdot  |U|^{\Ocal(1)}$,  where $s$ is the size of the input of the operation.
  Given two weighted partitions $\Acal$ and $\Bcal$, $\join(\Acal, \Bcal)$ can be computed in time $|\Acal| \cdot |\Bcal|\cdot |U|^{\Ocal(1)}$.
\end{proposition}

Given a weighted partition $\Acal \subseteq \Pi(U) \times \Nbb$ and a partition $q \in \Pi(U)$, we define
$\opt(q,\Acal) = \min \{w \mid (p,w) \in \Acal, p \sqcap q = \{U\}\}$.
Given two weighted partitions $\Acal,\Acal' \subseteq \Pi(U) \times \Nbb$, we say that $\Acal$ \emph{represents} $\Acal'$ if for each $q \in \Pi(U)$, $\opt(q,\Acal) = \opt(q,\Acal')$.

Given a set $Z$ and a function $f: 2^{\Pi(U) \times \Nbb} \times Z \to 2^{\Pi(U) \times \Nbb}$, we say that $f$ \emph{preserves representation} if for each two weighted partitions $\Acal,\Acal' \subseteq \Pi(U) \times \Nbb$ and each $z \in Z$, it holds that if $\Acal'$ represents $\Acal$  then $f(\Acal', z)$ represents $f(\Acal,z)$.

\begin{proposition}[Bodlaender et al.~\cite{BodlaenderCKN15}]
  \label{thorough}
  The union, insert, shift, glue, project, and join operations preserve representation.
\end{proposition}

\begin{theorem}[Bodlaender et al.~\cite{BodlaenderCKN15}]
  \label{happened}
  There exists an algorithm $\reduce$ that, given a set of weighted partitions $\Acal \subseteq \Pi(U) \times \Nbb$, outputs in time
  $|\Acal|\cdot 2^{(\omega-1) |U|}\cdot |U|^{\Ocal(1)}$ a set of weighted partitions $\Acal'\subseteq \Acal$ such that $\Acal'$ represents $\Acal$ and $|\Acal'| \leq 2^{|U|}$, where $\omega$ denotes the matrix multiplication exponent.
\end{theorem}

\medskip
\noindent
\textbf{Definition of the problems.} Let ${\cal F}$ be a proper collection.
We define the parameter ${\bf tm}_{\cal F}$ as the function that maps graphs to non-negative integers as follows:
\begin{eqnarray}
  \label{overlaid}
  {\bf tm}_{\cal F}(G) & = & \min\{|S|\mid  S\subseteq V(G)\wedge G\setminus S\in{\sf ex}_{\sf tm}({\cal F})\}.
\end{eqnarray}
The parameter ${\bf m}_{\cal F}$ is defined analogously.
The main objective of this paper is to study the problem of computing the parameters
${\bf tm}_{\cal F}$ and ${\bf m}_{\cal F}$  for graphs of bounded treewidth under several instantiations of
the collection ${\cal F}$. The corresponding decision problems are formally defined as follows.
\medskip\medskip\medskip

\hspace{-.6cm}\begin{minipage}{6.7cm}
  \paraprobl{\textsc{$\Fcal$-TM-Deletion}}
  {A graph $G$ and an integer $k\in \Nbb$.}
  {Is ${\bf tm}_\Fcal(G)\leq k$?}
  {The treewidth of $G$.}{6.7}
\end{minipage}~~~~~
\begin{minipage}{6.7cm}
  \paraprobl{\textsc{$\Fcal$-M-Deletion}}
  {A graph $G$ and an integer $k\in \Nbb$.}
  {Is ${\bf m}_\Fcal(G)\leq k$?}
  {The treewidth of $G$.}{6.7}
\end{minipage}
\medskip\medskip

Note that in both above problems,
we can always assume that ${\cal F}$ is an antichain
with respect to the considered relation.  Indeed, this is the case because if ${\cal F}$
contains two graphs $H_{1}$ and $H_{2}$ where $H_{1}\preceq_{\sf tm} H_{2}$, then
${\bf tm}_{\cal F}(G)={\bf tm}_{{\cal F}'}(G)$ where ${\cal F}'={\cal F}\setminus\{H_{2}\}$ (similarly
for the minor relation).

Throughout the article, we let $n$ and $\tw$ be the number of vertices and the treewidth of the input graph of the considered problem, respectively. We will also use $w$ to denote the width of a (nice) tree decomposition that is given together with the input graph (which, based on~\cite{BodlaenderDDFLP16}, will differ from $\tw$  by at most  a factor five).

\section{A single-exponential algorithm for {\sc $\{P_3\}$-TM-Deletion}}
\label{accuracy}

It should be noted that a single-exponential algorithm for {\sc $\{P_3\}$-TM-Deletion} is already known. Indeed, Tu et al.~\cite{P3-cover} presented an algorithm running in time $\Ostar(4^{\tw})$, and very recently Bai et al.~\cite{P3-cover-improved} improved it to $\Ostar(3^{\tw})$.
Nevertheless, for completeness we present in this section a simpler algorithm, but involving a greater constant than~\cite{P3-cover,P3-cover-improved}.

We first give a simple structural characterization of the graphs that exclude $P_3$ as a topological minor.


\begin{lemma}
  Let $G$ be a graph.
  $P_3 \not \pretp G$ if and only if each vertex of $G$ has degree at most one.
\end{lemma}
\begin{proof}
  Let $G$ be a graph.
  If $G$ has a connected component of size at least three, then clearly it contains a $P_3$.
  This implies that, if $P_3 \not \pretp G$, then  each connected component of $G$ has size at most two and so, each vertex of $G$ has degree at most one.
  Conversely, if each  vertex of $G$ has degree at most one, then, as $P_3$ contains a vertex of degree two, $P_3 \not \pretp G$.
\end{proof}

We present an algorithm using classical dynamic programming techniques over a tree decomposition of the input graph.
Let $G$ be an instance of {\sc $\{P_3\}$-TM-Deletion} and let $((T,\Xcal), r, \Gcal)$ be a nice tree decomposition of $G$.

We define, for each $t \in V(T)$,
the set
$\Ical_t = \{(S,S_0) \mid S,S_0 \subseteq X_t,\ S \cap S_0 = \es\}$ and a function $\dptable_t: \Ical_t \to \Nbb$ such that for each $(S,S_0) \in \Ical_t$, $\dptable(S,S_0)$ is the minimum $\ell$ such that
there exists a set $\hs \subseteq V(G_t)$, called the \emph{witness} of $(S,S_0)$, that satisfies:
\begin{itemize}
\item $|\hs| \leq \ell$,
\item $\hs \cap X_t = S$,
\item $P_3 \not \pretp G_t \gm \hs$, and
\item $S_0$ is the set of vertices of $X_t$ of degree $0$ in $G_t \gm S$.
\end{itemize}

Note that with this definition, 
${\bf tm}_{\cal F}(G) = \dptable_r(\es,\es)$.
For each $t \in V(T)$, we assume that we have already computed $\dptable_{t'}$ for each children $t'$ of $t$, and we proceed to the computation of $\dptable_t$.
We distinguish
several cases depending on the type of node $t$.

\begin{description}
\item[Leaf.] $\Ical_t = \{(\es,\es)\}$ and $\dptable_t(\es,\es) = 0$.
\end{description}

\begin{description}
\item[Introduce vertex.] If $v$ is the insertion vertex of $X_t$ and $t'$ is the child of $t$, then for each $(S,S_0) \in \Ical_t$,
  \begin{eqnarray*}
    \dptable_t(S,S_0) &=& \min\big(~~\{\dptable_{t'}(S',S_0)+1 \mid (S',S_0) \in \Ical_{t'},\ S = S' \cup \{v\}\}\\
                      &&~~~~~~\cup\ \{\dptable_{t'}(S,S_0') \mid (S,S_0' ) \in \Ical_{t'}, S_0 = S_0' \cup \{v\},\ N_{G_t[X_t]}(v) \sm S = \es\}
    \\
                      &&~~~~~~\cup\ \{\dptable_{t'}(S,S_0') \mid (S,S_0') \in \Ical_{t'}, S_0 = S_0' \sm \{u\},\ u \in S_0',\ \\
                      &&~~~~~~~~~~~~~~~~~~~~~~~~~~~~~~~~~~~~~~~~~~~~~~~~~~~~~~~N_{G_t[X_t]}(v) \sm S = \{u\}\}~\big).
  \end{eqnarray*}
\item[Forget vertex.] If  $v$ is the forget vertex of $X_t$ and $t'$ is the child of $t$, then for each $(S,S_0) \in \Ical_t$,
  \begin{eqnarray*}
    \dptable_t(S,S_0) &=& \min\{\dptable_{t'}(S', S_0')\mid (S',S_0')\in \Ical_{t'},\ S = S' \sm \{v\},\ S_0 = S_0' \sm \{v\}\}
  \end{eqnarray*}

\item[Join.]  If $t'$ and $t''$ are the children of $t$, then for each $(S,S_0) \in \Ical_t$,
  \begin{eqnarray*}
    \dptable(S,S_0) &=& \min\{\dptable(S',S_0') + \dptable(S'',S_0'') - |S' \cap S''| \\
                    &&~~~~~~~~\mid
                       (S',S_0') \in \Ical_{t'}, (S'',S_0'') \in \Ical_{t''}, \\
                    &&~~~~~~~~~~~~~~~~~
                       S = S' \cup S'',\
                       S_0 = S_0' \cap S_0'',\
                       X_t \sm S \subseteq S_0' \cup S_0''
                       \}.
  \end{eqnarray*}
\end{description}

Let us analyze the running time of this algorithm.
As, for each $t \in V(T)$, $S$ and $S_0$ are disjoint subsets of $X_t$, we have that $|\Ical_t| \leq 3^{|X_t|}$.
Note that
if $t$ is a leaf, then $\dptable_t$ can be computed in time $\Ocal(1)$,
if $t$ is an introduce vertex or a forget vertex node, and $t'$ is the child of $t$, then $\dptable_t$ can be computed in time $\Ocal(|\Ical_{t'}|\cdot |X_t|)$,
and if $t$ is a join node, and $t'$ and $t''$ are the two children of $t$, then  $\dptable_t$ can be computed in time $\Ocal(|\Ical_{t'}|\cdot |\Ical_{t''}| \cdot |X_t|)$.

We now show that for each $t \in V(T)$, the function $\dptable_t$ is correctly computed by the algorithm.

\begin{description}
\item[Leaf.] This follows directly from the definition of $\dptable_t$.
\item[Introduce vertex.] Let $v$ be the insertion vertex of $X_t$.
  As $v$ is the insertion vertex, we have that $N_{G_t[X_t]}(v) = N_{G_t}(v)$, and so
  for each value we add to the set, we can find a witness of $(S,S_0)$ of size bounded by this value.

  Conversely, let $(S, S_0) \in \Ical_t$ and let $\hs$ be a witness.
  If $v \in S$, then $(S\sm \{v\}, S_0) \in \Ical_{t'}$ and $\dptable(S\sm \{v\}, S_0) \leq |\hs|-1$,
  if $v \in S_0$ then $(S, S_0\sm \{v\}) \in \Ical_{t'}$ and $\dptable(S, S_0\sm \{v\}) \leq |\hs|$, and
  if $v \in X_t \sm (S \cup S_0)$, then by definition $v$ has a unique neighbor, say $u$, in $G_t \gm \hs$,
  moreover $u \in X_t \sm (S \cup S_0)$, $v$ is the unique neighbor of $u$ in $G_t \gm \hs$,  $(S, S_0\cup \{u\}) \in \Ical_{t'}$, and $\dptable(S, S_0\cup \{u\}) \leq |\hs|$.





\item[Forget vertex.] This also follows directly from the definition of $\dptable_t$.

\item[Join.]  Let $(S',S_0') \in \Rcal_{t'}$ and let   $(S'',S_0'') \in \Ical_{t''}$ with witnesses $\hs'$ and $\hs''$, respectively.
  If $S = S' \cup S''$ and $S_0' \cup S_0'' = X_t \sm S$, then
  the condition $ X_t \sm S \subseteq S_0' \cup S_0''$ ensures that $G_t \gm (\hs' \cup \hs'')$ has no vertex of degree at least two  and so
  $\hs' \cup \hs''$ is a witness of $(S,S_0'\cap S_0'') \in \Ical_t$ of size at most $\dptable_{t'}(S',S_0')+ \dptable_{t'}(S'',S_0'') -|S'\cap S''|$.

  Conversely,
  let  $(S, S_0) \in \Ical_t$ with witness $\hs$.
  If $\hs' = \hs \cap V(G_{t'})$ and $\hs'' = \hs \cap V(G_{t''})$, then by definition of $\hs$,
  $\hs'$
  is a witness of  some $(S',S_0') \in \Ical_{t'}$, and
  $\hs''$ is a witness of some $(S'',S_0'') \in \Ical_{t''}$
  such that
  $S = S' = S''$, $S_0' \cup S_0'' = X_t \sm S$, and $S_0 = S_0'\cap S_0''$, and we have
  $\dptable_{t'}(S',S_0')+ \dptable_{t'}(S'',S_0'') -|S| \leq |\hs|$.


\end{description}

The following theorem summarizes the above discussion. 

\begin{theorem}\label{majority}
  If a nice tree decomposition of $G$ of width $w$ is given, {\sc $\{P_3\}$-TM-Deletion} can be solved in time $\Ocal(9^{w}\cdot w \cdot  n)$.
\end{theorem}

\section{A single-exponential algorithm for {\sc $\{P_4\}$-TM-Deletion}}
\label{affected}

Similarly to what we did for {\sc $\{P_3\}$-TM-Deletion}, we start with a structural definition of the graphs that exclude $P_4$ as a topological minor.

\begin{lemma}\label{headlong}
  Let $G$ be a graph.
  $P_4 \not \pretp G$ if and only if
  each connected component of $G$ is either a $C_3$ or a star.
\end{lemma}

\begin{proof}
  First note that if each connected component of $G$ is either a $C_3$ or a star, then $P_4 \not \pretp G$. Conversely, assume that $P_4 \not \pretp G$.
  Then each connected component of $G$ of size at least $4$ should contain at most $1$ vertex of degree at least $2$, hence such component is a star. On the other hand, the only graph on at most 3 vertices that is not a star is $C_3$. The lemma follows.
\end{proof}

As we did for {\sc $\{P_3\}$-TM-Deletion}, we present an algorithm using classical dynamic programming techniques over a tree decomposition of the input graph.
Let $G$ be an instance of {\sc $\{P_4\}$-Deletion}, and let $((T,\Xcal), r, \Gcal)$ be a nice tree decomposition of $G$.

We define, for each $t \in T$, the set $\Ical_t$ to be the set of each tuple $(S,S_{1+}, S_{1-},S_*,S_{3+}, S_{3-})$
such that $ \{S, S_{1+}, S_{1-}, S_*, S_{3+}, S_{3-}\}$ is a partition of $X_t$ and the function
$\dptable_t: \Ical_t \to \Nbb$ such that, for each $(S,S_{1+}, S_{1-},S_*,S_{3+}, S_{3-}) \in \Ical_t$,
$\dptable_t(S,S_{1+}, S_{1-},S_*,S_{3+}, S_{3-})$ is the minimum $\ell$ such that there exists  a triple $(\hs,\hs_*,\hs_{3-})\ \subseteq V(G_t) \times V(G_t) \times V(G_t)$,
called the \emph{witness} of $(S,S_{1+}, S_{1-},S_*,S_{3+}, S_{3-})$, which satisfies the following properties:

\begin{itemize}
\item $\hs$, $\hs_*$, and $\hs_{3-}$ are pairwise disjoint,
\item  $\widehat{S} \cap X_t = S$,
  $\widehat{S}_* \cap X_t = S_*$, and
  $\widehat{S}_{3-} \cap X_t = S_{3-}$,
\item $|\widehat{S}| \leq \ell$,
\item  $P_4 \not \pretp G_t \gm \hs$,
\item $S_{1+}$ is a set of vertices of degree $0$ in $G_t \gm \hs$,
\item each vertex of $S_{1-}$ has a unique neighbor in $G_t \gm \widehat{S}$ and this neighbor  is in $\widehat{S}_*$,
\item each connected component of $G_t[\widehat{S}_{3-}]$ is a $C_3$,
\item there is no edge in $G_t \gm \hs$ between a vertex of $\widehat{S}_{3-}$ and a vertex of $V(G_t) \sm (\hs \cup \widehat{S}_{3-})$,
\item there is no edge in $G_t \gm \hs$ between a vertex of ${S}_{3+}$ and a vertex of $V(G_t) \sm (\hs \cup {S}_{3+})$, and
\item there is no edge in $G_t \gm \hs$ between two vertices of $S_*$.
\end{itemize}

Intuitively, $\hs$ corresponds to a partial solution in $G_t$. Note that, by Lemma~\ref{headlong}, each component of $G_t \gm \hs$ must be either a star or a $C_3$. With this in mind, $\hs_*$ is the set of vertices that are centers of a star in $G_t \gm \hs$, $S_{1+}$ is the set of
leaves of a star
that are not yet connected to a vertex of $\hs_*$, $S_{1-}$ is the set of leaves of a star that are already connected to a vertex of $\hs_*$,  $\hs_{3-}$ is the set of vertices that induce $C_3$'s in $G_{t}$, and $S_{3+}$ is a set of vertices that will induce $C_3$'s when further edges will appear.


Note that with this definition, 
${\bf tm}_{\cal F}(G) =  \dptable_r(\es,\es,\es,\es,\es,\es)$.
For each $t \in V(T)$, we assume that we have already computed $\dptable_{t'}$ for each children $t'$ of $t$, and we proceed to the computation of $\dptable_t$.
We distinguish several cases depending on the type of node $t$.
\begin{description}
\item[Leaf.] $\Ical_t = \{(\es,\es,\es,\es,\es,\es)\}$ and $\dptable_t(\es,\es,\es,\es,\es,\es) = 0$.
\item[Introduce vertex.] If $v$ is the insertion vertex of $X_t$ and $t'$ is the child of $t$, then, for each $(S,S_{1+}, S_{1-},S_*,S_{3+}, S_{3-}) \in \Ical_t$,

  \begin{eqnarray*}
    \!\!\!\!\!\!\!\!\!\!\!\!\!\!\dptable_t(S,S_{1+},S_{1-},S_*,S_{3+},S_{3-}) &=&\min\big(~~\{
                                                                                  \dptable_{t'}(S',S_{1+},S_{1-},S_*,S_{3+},S_{3-})+1 \\
                                                                              &&~~~\mid (S',S_{1+},S_{1-},S_*,S_{3+},S_{3-}) \in \Rcal_{t'},\ S = S' \cup \{v\}
                                                                                 \}\\
                                                                              &\cup& \{\dptable_{t'}(S,S'_{1+},S_{1-},S_*,S_{3+},S_{3-}) \\
                                                                              &&~~~\mid (S,S'_{1+},S_{1-},S_*,S_{3+},S_{3-}) \in \Rcal_{t'},\\
                                                                              &&~~~~~ \ S _{1+}= S'_{1+} \cup \{v\},\
                                                                                 N_{G_t[X_t\sm S]}(v) = \es
                                                                                 \}\\
                                                                              &\cup& \{\dptable_{t'}(S,S_{1+},S'_{1-},S_*,S_{3+},S_{3-}) \\
                                                                              &&~~~\mid (S,S_{1+},S'_{1-},S_*,S_{3+},S_{3-}) \in \Rcal_{t'},\\
                                                                              &&~~~~~ \ S _{1-}= S'_{1-} \cup \{v\},\
                                                                                 z \in S_*,\ N_{G_t[X_t\sm S]}(v) = \{z\}
                                                                                 \}\\
                                                                              &\cup& \{\dptable_{t'}(S,S'_{1+},S'_{1-},S'_*,S_{3+},S_{3-}) \\
                                                                              &&~~~\mid (S,S'_{1+},S'_{1-},S'_*,S_{3+},S_{3-}) \in \Rcal_{t'},\\
                                                                              &&~~~~~ S_{*}= S'_{*} \cup \{v\},\
                                                                                 N_{G_t[X_t \sm S]}(v) \subseteq S'_{1+},\\
                                                                              &&~~~~~                S_{1+} = S'_{1+} \sm N_{G_t[X_t \sm S]}(v),\
                                                                                 S_{1-} = S'_{1-} \cup N_{G_t[X_t \sm S]}(v)
                                                                                 \}\\
                                                                              &\cup& \{\dptable_{t'}(S,S_{1+},S_{1-},S_*,S'_{3+},S_{3-}) \\
                                                                              &&~~~\mid (S,S_{1+},S_{1-},S_*,S'_{3+},S_{3-}) \in \Rcal_{t'},\\
                                                                              &&~~~~~ \ S_{3+}= S'_{3+} \cup \{v\},\ \\
                                                                              &&~~~~~~ [ N_{G_t[X_t \sm S]}(v) = \es] \mbox{ or }\\
                                                                              &&~~~~~~                [z \in S'_{3+}, \
                                                                                 N_{G_t[X_t \sm S]}(v) = \{z\},\
                                                                                 N_{G_t[X_t \sm S]}(z) = \{v\}]
                                                                                 \}\\
                                                                              &\cup& \{\dptable_{t'}(S,S_{1+},S_{1-},S_*,S'_{3+},S'_{3-}) \\
                                                                              &&~~~\mid (S,S_{1+},S_{1-},S_*,S'_{3+},S_{3-}) \in \Rcal_{t'},\\
                                                                              &&~~~~~~ S_{3+}= S'_{3+} \sm \{z,z'\},\ S_{3-}= S'_{3-} \cup \{z,z',v\},\ \\
                                                                              &&~~~~~~   z,z' \in S'_{3+}, \
                                                                                 N_{G_t[X_t \sm S]}(v) = \{z,z'\},\\
                                                                              &&~~~~~~
                                                                                 N_{G_t[X_t \sm S]}(z) = \{v,z'\},\
                                                                                 N_{G_t[X_t \sm S]}(z') = \{v,z\}
                                                                                 \}~~\big).\\
  \end{eqnarray*}
\item[Forget vertex.] If  $v$ is the forget vertex of $X_t$ and $t'$ is the child of $t$, then,  \\ for each $(S,S_{1+}, S_{1-},S_*,S_{3+}, S_{3-}) \in \Ical_t$,
  \begin{eqnarray*}
    \dptable_t(S,S_{1+}, S_{1-},S_*,S_{3+}, S_{3-}) &=& \min\{\dptable_{t'}(S',S_{1+}, S'_{1-},S'_*,S_{3+}, S'_{3-})\\
                                                    &&~~~~~~~~~~\mid (S',S_{1+},S'_{1-},S'_*,S_{3+},S'_{3-}) \in \Ical_{t'},\ \\
                                                    &&~~~~~~~~~~~~ S = S' \sm \{v\},\ S_{1-} = S'_{1-} \sm \{v\},\\
                                                    &&~~~~~~~~~~~~ S_* = S'_* \sm \{v\},\ S_{3-} = S'_{3-} \sm \{v\}\}.
  \end{eqnarray*}
\item[Join.]  If $t'$ and $t''$ are the children of $t$, then for each $(S,S_{1+}, S_{1-},S_*,S_{3+}, S_{3-}) \in \Ical_t$, \\ $\dptable_t(S,S_{1+}, S_{1-},S_*,S_{3+}, S_{3-})$ is
  \begin{eqnarray*}
    && \min\{\dptable_{t'}(S,S'_{1+}, S'_{1-},S_*,S'_{3+}, S'_{3-}) + \dptable_{t'}(S,S''_{1+}, S''_{1-},S_*,S''_{3+}, S''_{3-}) - |S|\\
    &&~~~~~~~~~~\mid (S,S'_{1+},S'_{1-},S_*,S'_{3+},S'_{3-}) \in \Ical_{t'},\ (S,S''_{1+},S''_{1-},S_*,S''_{3+},S''_{3-}) \in \Ical_{t''},\ \\
    &&~~~~~~~~~~ (S'_{1+} \cup S'_{1-}) \cap (S''_{3+} \cup S''_{3-}) = (S''_{1+} \cup S''_{1-}) \cap (S'_{3+} \cup S'_{3-}) = \es,\\
    &&~~~~~~~~~~ \forall v \in S'_{1-}\cap S''_{1-},\ \exists z \in S_* : N_{G_t[X_t \sm S]}(v) = \{z\},\\
    &&~~~~~~~~~~ S_{1-} = (S'_{1-}\cup S''_{1-}),\ S_{1+} = S'_{1+}\cap S''_{1+},\\
    &&~~~~~~~~~~ \forall v \in S'_{3-}\cap S''_{3-}, \exists z,z' \in S'_{3-}\cap S''_{3-} : v,z,z'\ \mbox{induce a $C_3$\! in\! $G_t[{X_t \sm S}]$},\\
        &&~~~~~~~~~~ S_{3-} = (S'_{3-}\cup S''_{3-}) ,\ S_{3+} = S'_{3+}\cap S''_{3+}
       \}.\\
  \end{eqnarray*}
\end{description}
Let us analyze  the running time of this algorithm.
As, for each $t \in V(T)$, $S$, $S_{1+}$, $S_{1-}$, $S_*$, $S_{3+}$, and $S_{3-}$ form a partition of $X_t$, we have that $|\Ical_t| \leq 6^{|X_t|}$.
Note that
if $t$ is a leaf, then $\dptable_t$ can be computed in time $\Ocal(1)$,
if $t$ is an introduce vertex or a forget vertex node, and $t'$ is the child of $t$, then $\dptable_t$ can be computed in time $\Ocal(|\Ical_{t'}|\cdot |X_t|)$, and
if $t$ is a join node, and $t'$ and $t''$ are the two children of $t$, then  $\dptable_t$ can be computed in time $\Ocal(|\Ical_{t'}|\cdot |\Ical_{t''}|\cdot |X_t|)$.

We now show that for each $t \in V(T)$, $\dptable_t$ is correctly computed by the algorithm.
For each  $(S,S_{1+}, S_{1-},S_*,S_{3+}, S_{3-})\in \Ical_t$, it can be easily checked that each value $\ell$ we compute
respects,  $\dptable_t(S,S_{1+}, S_{1-},S_*,S_{3+}, S_{3-}) \leq \ell$.
Conversely, we now argue that
for each  $(S,S_{1+}, S_{1-},S_*,S_{3+}, S_{3-})\in \Ical_t$, the computed value $\ell$ is such that
each witness $(\hs,\hs_*, \hs_{3-})$ of $(S,S_{1+}, S_{1-},S_*,S_{3+}, S_{3-})$ satisfies $\ell \leq |\hs|$.
We again distinguish the type of node $t$.

\begin{description}
\item[Leaf.] This follows directly from the definition of $\dptable_t$.
\item[Introduce vertex.]
  Let $v$ be the insertion vertex of $X_t$, let  $(S,S_{1+}, S_{1-}, S_*, S_{3+}, S_{3-}) \in \Rcal_t$, and let $(\hs,\hs_*, \hs_{3-})$ be a witness.
  \begin{itemize}
  \item If $v \in S$, then $(S \sm \{v\},S_{1+},S_{1-},S_*,S_{3+},S_{3-}) \in \Ical_{t'}$ and \\$\dptable_{t'}(S \sm \{v\},S_{1+},S_{1-},S_*,S_{3+},S_{3-}) \leq |\hs|-1$.
  \item If $v \in S_{1+}$, then $v$ is of degree $0$ in $G_t \gm \hs$,
    hence  $(S ,S_{1+}\sm \{v\},S_{1-},S_*,S_{3+},S_{3-}) \in \Ical_{t'}$ and $\dptable_{t'}(S ,S_{1+}\sm \{v\},S_{1-},S_*,S_{3+},S_{3-}) \leq |\hs|$.
  \item If $v \in S_{1-}$, then $v$ has a unique neighbor that is in $\hs_*$. As $v$ is the insertion vertex of $X_t$, it implies that $N_{G_t}(v) \subseteq S_*$, and so $(S ,S_{1+},S_{1-} \sm \{v\},S_*,S_{3+},S_{3-}) \in \Ical_{t'}$ and
    $\dptable_{t'}(S ,S_{1+},S_{1-}\sm \{v\},S_*,S_{3+},S_{3-}) \leq |\hs|$.
  \item If $v \in S_*$, then every neighbor of $v$ is in $S_{1-}$ and has degree $1$ in $G_t \gm \hs$. Thus,
    $(S,S_{1+}\cup N_{G_t[X_t\sm S]}(v),S_{1-} \sm  N_{G_t[X_t\sm S]}(v),S_* \sm \{v\},S_{3+},S_{3-}) \in \Ical_{t'}$ and
    $\dptable_{t'}(S,S_{1+}\cup N_{G_t[X_t\sm S]}(v),S_{1-} \sm  N_{G_t[X_t\sm S]}(v),S_* \sm \{v\},S_{3+},S_{3-}) \leq |\hs|$.
  \item If $v \in S_{3+}$, then $(S ,S_{1+},S_{1-},S_*,S_{3+} \sm \{v\},S_{3-}) \in \Ical_{t'}$ and

    $\dptable_{t'}(S ,S_{1+},S_{1-},S_*,S_{3+} \sm \{v\},S_{3-}) \leq |\hs|$.
  \item If $v \in S_{3-}$, then there exist $z$ and $z'$ in $S_{3-}$ such that $\{v,z,z'\}$ induce a $C_3$ in $G_t \gm \hs$ and there is no edge in $G_t \gm \hs$ between a vertex of $\{v,z,z'\}$ and a vertex of $V(G_t \gm \hs) \sm \{x,z,z'\}$.
    So  $(S ,S_{1+},S_{1-},S_*,S_{3+} \cup \{z,z'\},S_{3-}\sm \{x,z,z'\}) \in \Ical_{t'}$
    and $\dptable_{t'}(S ,S_{1+},S_{1-},S_*,S_{3+} \cup \{z,z'\},S_{3-}\sm \{x,z,z'\}) \leq |\hs|$.
  \end{itemize}
\item[Forget vertex.]
  Let $v$ be the forget vertex of $X_t$, let  $(S,S_{1+}, S_{1-}, S_*, S_{3+}, S_{3-}) \in \Ical_t$, and let $(\hs,\hs_*, \hs_{3-})$ be a witness.
  If $v$ has degree $0$ in $G_t \gm \hs$, then
  $(S,S_{1+}, S_{1-}, S_* \cup \{v\}, S_{3+}, S_{3-}) \in \Ical_{t'}$ and
  $\dptable_{t'}(S,S_{1+}, S_{1-}, S_* \cup \{v\}, S_{3+}, S_{3-}) \leq |\hs|$.
  If $v$ has degree at least $1$ in $G_t \gm \hs$, then $N_{G_t \gm \hs}(v) \cap  S_{3+} = \es$, as otherwise there would be an edge in $G_t \gm \hs$ between a vertex of $S_{3+}$ and a vertex of $V(G_t) \sm (\hs \cup {S}_{3+})$. So, one of the following case occurs:
  \begin{itemize}
  \item $v \in \hs$, $(S \cup \{v\},S_{1+}, S_{1-}, S_*, S_{3+}, S_{3-}) \in \Ical_{t'}$, and \\$\dptable_{t'}(S \cup \{v\},S_{1+}, S_{1-}, S_*, S_{3+}, S_{3-}) \leq |\hs|$,
  \item $v \in \hs_*$,  $(S,S_{1+}, S_{1-}, S_* \cup \{v\}, S_{3+}, S_{3-}) \in \Ical_{t'}$, and\\ $\dptable_{t'}(S,S_{1+}, S_{1-}, S_* \cup \{v\}, S_{3+}, S_{3-}) \leq |\hs|$,
  \item $N_{G_t \gm \hs}(v) \subseteq \hs_*$, $(S,S_{1+}, S_{1-}\cup \{v\}, S_* , S_{3+}, S_{3-}) \in \Ical_{t'}$, and\\ $\dptable_{t'}(S,S_{1+}, S_{1-} \cup \{v\}, S_*, S_{3+}, S_{3-}) \leq |\hs|$, or
  \item $v \in \hs_{3-}$, $(S,S_{1+}, S_{1-}, S_* , S_{3+}, S_{3-} \cup \{v\}) \in \Ical_{t'}$, and\\ $\dptable_{t'}(S,S_{1+}, S_{1-}, S_*, S_{3+}, S_{3-} \cup \{v\}) \leq |\hs|$
  \end{itemize}

\item[Join.]
  Let  $(S,S_{1+}, S_{1-}, S_*, S_{3+}, S_{3-}) \in \Ical_t$, and let $(\hs,\hs_*, \hs_{3-})$ be a witness.
  Let $t'$ and $t''$ be the two children of $t$.
  We define
  $\hs' = \hs \cap V(G_{t'})$,
  $\hs'' = \hs \cap V(G_{t''})$,
  $\hs'_* = \hs_* \cap V(G_{t'})$,
  $\hs''_* = \hs_* \cap V(G_{t''})$,
  $\hs'_{3-} \subseteq \hs_{3-} \cap V(G_{t'})$, and
  $\hs''_{3-} \subseteq \hs_{3-} \cap V(G_{t''})$,
  such that each connected component of $G_t[\hs'_{3-}]$ (resp. $G_t[\hs''_{3-}]$) is a $C_3$ and
  $G_{t'} \gm (\hs'\cup \hs'_{3-})$ (resp. $G_{t''} \gm (\hs'' \cup \hs''_{3-})$) is a forest).
  Then we define
  \begin{itemize}
  \item $S' = \hs' \cap X_t$,
  \item $S'_{1+} = S_{1+} \cup \{v \in S_{1-} \mid N_{G_t \gm \hs}(v) \not \subseteq \hs'_*\}$,
  \item $S'_{1-} = \{v \in S_{1-} \mid N_{G_t \gm \hs}(v) \subseteq \hs'_*\}$,
  \item $S'_* = S_* \cap V(G_{t'})$,
  \item $S'_{3-} = \hs_{3-}' \cap X_t$, and
  \item $S'_{3+} = S_{3+} \cup (S_{3-}\sm S'_{3-})$.
  \end{itemize}

  Note that
  $(S' ,S'_{1+},S'_{1-},S'_*,S'_{3+},S'_{3-}) \in \Ical_t'$.
  We define
  $(S'' ,S''_{1+},S''_{1-},S''_*,S''_{3+},S''_{3-}) \in \Ical_t''$ similarly.
  Moreover we can easily check that
  \begin{itemize}
  \item $ S = S' =  S'', S_* = S'_* = S''_*$,
  \item     $(S'_{1+} \cup S'_{1-}) \cap (S''_{3+} \cup S''_{3-}) = (S''_{1+} \cup S''_{1-}) \cap (S'_{3+} \cup S'_{3-}) = \es$,
  \item    $\forall v \in S'_{1-}\cap S''_{1-}, \exists z \in S_* : N_{G_t[X_t \sm S]}(v) = \{z\}$,
  \item     $\forall v \in S'_{3-}\cap S''_{3-}, \exists z,z' \in S'_{3-}\cap S''_{3-} : v,z,z' \mbox{ induce a $C_3$ in $G_t[{X_t \sm S}]$}$,
  \item $(S,S_{1+}, S_{1-}, S_*, S_{3+}, S_{3-}) = (S,S'_{1+}\cap S''_{1+},S'_{1-}\cup S''_{1-},S_*,S'_{3+} \cap S''_{3+},S'_{3-}\cup S''_{3-})$, and
  \item $\dptable_{t'}(S' ,S'_{1+},S'_{1-},S'_*,S'_{3+},S'_{3-}) + \dptable_{t''}(S'' ,S''_{1+},S''_{1-},S''_*,S''_{3+},S''_{3-}) - |S| \leq |\hs|$.
  \end{itemize}
\end{description}

This concludes the proof of correctness of the algorithm. The following theorem summarizes the above discussion.
\begin{theorem}
  If a nice tree decomposition of $G$ of width $w$ is given, {\sc $\{P_4\}$-Deletion} can be solved in time $\Ocal(36^{w}\cdot w \cdot n)$.
\end{theorem}

\section{Single-exponential algorithms for $\{K_{1,s}\}$-{\sc TM-Deletion}}
\label{repaired}

Similarly to what we did before, we start with a (trivial) structural characterization of the graphs that exclude $K_{1,s}$, for some fixed integer $s$, as a topological minor.


\begin{lemma}
  \label{observer}
  Let $s$ be a positive integer. A graph $G$ contains $K_{1,s}$ as a topological minor if and only if it contains a vertex of degree at least $s$.
\end{lemma}
\begin{proof}
  Let $s$ be a fixed integer.
  If a graph $G$ contains a vertex $v$ of degree at least $s$, then $G$ contains $K_{1,s}$ as a subgraph and so, as a topological minor.
  If $G$ contains $K_{1,s}$ as a topological minor, then it implies that there exist in $G$ a vertex $v$ and $s$ paths of size at least two such that the intersection of any two of these paths contains precisely $v$.
  Thus $v$ has degree at least $s$.
\end{proof}

Given a fixed integer $s \geq 1$, the \textsc{$s$-Bounded-degree Vertex Deletion} problem asks, given a graph $G$ and an integer $k$, whether one can remove at most $k$ vertices from $G$ such that the remaining graph has maximum degree at most $s$.
Lemma~\ref{observer} implies that for every positive integer $s$, \textsc{$\{K_{1,s}\}$-TM-Deletion} is exactly \textsc{$(s-1)$-Bounded-degree Vertex Deletion}.
For completeness, we provide a simple single-exponential algorithm parameterized by treewidth that solves \textsc{$s$-Bounded-degree Vertex Deletion} for any fixed integer $s \geq 1$.

Let $s \geq 1$ be a fixed integer, let $G$ be an instance of \textsc{$s$-Bounded-degree Vertex Deletion}, and
let $((T,\Xcal), r, \Gcal)$ be a nice tree decomposition of $G$.
We define, for each $t \in V(T)$, the set $\Ical_t = \{(S,f) \mid S \subseteq X_t, f: X_t \setminus S \to \intv{0,s-1}\}$ and a function
$\dptable_t: \Ical_t \to \Nbb$ such that for each $(S,f) \in \Ical_t$, $\dptable(S,f)$ is the minimum $\ell$ such that
there exists a set $\hs \subseteq V(G_t)$, called the \emph{witness} of $(S,f)$, that satisfies:
\begin{itemize}
\item $|\hs| \leq \ell$,
\item $\hs \cap X_t = S$, and
\item for each $v \in X_t \setminus S$, $\degs{G_t \sm \hs}{v} = f(v)$.
\end{itemize}
Note that with this definition,
${\bf tm}_{\cal F}(G) = \dptable_r(\es,\varnothing)$.
For each $t \in V(T)$, we assume that we have already computed $\dptable_{t'}$ for each children $t'$ of $t$, and we proceed to the computation of $\dptable_t$.
We distinguish
several cases depending on the type of node $t$.

\begin{description}
\item[Leaf.] $\Ical_t = \{(\es,\ef)\}$ and $\dptable_t(\es,\ef) = 0$.
\item[Introduce vertex.] If $v$ is the insertion vertex of $X_t$ and $t'$ is the child of $t$, then for each $(S,f) \in \Ical_t$,
  \begin{eqnarray*}
    \dptable_t(S,f) &=& \min\big(~~\{\dptable_{t'}(S',f)+1 \mid (S',f) \in \Ical_{t'},\ S = S' \cup \{v\}\}\\
                    &&~~~~~~\cup\ \{\dptable_{t'}(S,f') \mid (S,f') \in \Ical_{t'},
                       f(v) = \degs{G[X_t\sm S]}{v},\\
                    &&~~~~~~~~~~~~~~~~~~~~~~~~~\forall v' \in N_{G_t[X_{t}\sm S]}(v),\ f(v') = f'(v')+1,\\
                    &&~~~~~~~~~~~~~~~~~~~~~~~~~\forall v' \in X_{t'}\sm (S \cup N_{G_t[X_{t}\sm S]}(v)),\ f(v') = f'(v')
                       \}~\big).
  \end{eqnarray*}
\item[Forget vertex.] If  $v$ is the forget vertex of $X_t$ and $t'$ is the child of $t$, then for each $(S,f) \in \Ical_t$,
  \begin{eqnarray*}
    \dptable_t(S,f) &=& \min\{\dptable_{t'}(S', f')\mid (S',f')\in \Ical_{t'},\ S = S' \sm \{v\},\ \forall v' \in X_t\sm S,\ f(v') = f'(v')\}.
  \end{eqnarray*}

\item[Join.]  If $t'$ and $t''$ are the children of $t$, then for each $(S,f) \in \Ical_t$,
  \begin{eqnarray*}
    \dptable(S,f) &=& \min\{\dptable(S,f') + \dptable(S,f'') - |S| \\
                  &&~~~~~~~~\mid
                     (S,f') \in \Ical_{t'}, (S,f'') \in \Ical_{t''},\\
                  &&~~~~~~~~~~
                     \forall v \in X_t \sm S,\ f(v) = f'(v) + f''(v) - \degs{G_t[X_t\sm S]}{v}
                     \}.
  \end{eqnarray*}
\end{description}
One can check that for each $t \in V(T)$, the set $\Ical_t$ is of size at most $(s+1)^{|X_t|}$: for each vertex in $X_t \setminus S$ there are $s$ possible values for its degree, together with the choice of belonging to $S$ or not for each vertex in $X_t$. Using the same argumentation as in the previous algorithms, we obtain the following theorem.

\begin{theorem}\label{confined}
  Let $s \geq 1$ be a fixed  integer.
  If a nice tree decomposition of $G$ of width $w$ is given, {\sc $\{K_{1,s}\}$-TM-Deletion} can be solved in time $\Ocal((s+1)^{2w}\cdot w \cdot  n)$.
\end{theorem}

\section{A single-exponential algorithm for {\sc $\{C_4\}$-TM-Deletion}}
\label{benjamin}

As discussed before, in this section we use the dynamic programming techniques introduced by Bodlaender et al.~\cite{BodlaenderCKN15} to obtain a single-exponential algorithm for {\sc $\{C_4\}$-TM-Deletion}. It is worth mentioning that the \textsc{$\{C_i\}$-TM-Deletion} problem has been studied in digraphs from a non-parameterized point of view~\cite{PaikRS94}.
The algorithm we present solves the decision version of {\sc $\{C_4\}$-TM-Deletion}: the input is a pair $(G,k)$, where $G$ is a graph and $k$ is an integer, and the output is the boolean value ${\bf tm}_\Fcal(G) \leq k$.

We give some definitions that will be used for the following algorithm.
Given a graph $G$, we denote by
$n(G) = |V(G)|$,
$m(G) = |E(G)|$,
$\ct(G)$ the number of $C_3$'s that are subgraphs of $G$, and
$\cc(G)$ the number of connected components of $G$.
We say that $G$ satisfies the \emph{$C_4$-condition} if the following conditions hold:
\begin{itemize}
\item $G$ does not contain the \ourdiamond
  as a subgraph, and
\item $n(G) - m(G) + \ct(G) = \cc(G)$.
\end{itemize}

As in the case of $P_3$ and $P_4$, we state in Lemma~\ref{lifeboat} a structural characterization of the graphs that exclude $C_4$ as a (topological) minor. We first need an auxiliary lemma.

\begin{lemma}
  \label{brawling}
  Let $n_0$ be a positive integer.
  Assume that for each graph $G'$ such that $1 \leq n(G') \leq n_0$,
  $C_4 \not \pretp G'$ if and only if $G'$ satisfies the $C_4$-condition.
  If $G$ is a graph that does not contain a \ourdiamond as a subgraph and such that $n(G) = n_0$, then
  $n(G) - m(G) + \ct(G) \leq \cc(G)$.
\end{lemma}
\begin{proof}
  Let $n_0$ be a positive integer, and
  assume that for each graph $G'$ such that $1 \leq n(G') \leq n_0$,
  $C_4 \not \pretp G'$ if and only if $G$ satisfies the $C_4$-condition.
  Let $G$ be a graph that does not contain a \ourdiamond as a subgraph and such that $n(G) = n_0$.
  Let $S \subseteq E(G)$ such that $C_4 \not \pretp G \gm S$ and  $\cc(G\gm S) = \cc(G)$ (note that any minimal feedback edge set satisfies these conditions).
  We have, by hypothesis, that $G \gm S$ satisfies the $C_4$-condition, so $n(G\gm S) - m(G\gm S) + \ct(G\gm S) = \cc(G\gm S)$.
  Moreover, as $G$ does not contain a \ourdiamond as a subgraph, each edge of $G$ participates in at most one $C_3$, and thus $\ct(G) - \ct(G \gm S) \leq |S|$.
  As by definition $n(G) = n(G \gm S)$ and $m(G) - m(G \gm S) = |S|$, we obtain that $n(G) - m(G) + \ct(G) \leq \cc(G\gm S) =  \cc(G)$.
\end{proof}

\begin{lemma}
  \label{lifeboat}
  Let $G$ be a non-empty graph.
  $C_4 \not \pretp G$ if and only if $G$ satisfies the $C_4$-condition.
\end{lemma}
\begin{proof}
  Let $G$ be a non-empty  graph, and assume first that $C_4 \not \pretp G$.
  This directly implies that $G$ does not contain the \ourdiamond as a subgraph.
  In particular, any two cycles of $G$, which are necessarily $C_3$'s, cannot share an edge.
  Let $S$ be a set containing an arbitrary edge of each $C_3$  in $G$.
  By construction, $G \gm S$ is a forest.
  As in a forest $F$, we have $n(F)-m(F)=\cc(F)$, and $S$ is defined such that $|S| = \ct(G)$ because each edge of $G$ participates in at most one $C_3$, we obtain that $n(G)-m(G)+\ct(G)=\cc(G)$.
  Thus, $G$ satisfies the $C_4$-condition.

  Conversely, assume now that $G$ satisfies the $C_4$-condition.
  We prove that $C_4 \not \pretp G$ by induction on $n(G)$.
  If $n(G) \leq 3$, then $n(G) < n(C_4)$ and so $C_4 \not \pretp G$.
  Assume now that $n(G) \geq 4$, and that for each graph  $G'$ such that $1 \leq n(G') < n(G)$,
  if $G'$  satisfies the $C_4$-condition, then $C_4 \not \pretp G'$.
  We prove that this last implication is also true for  $G$.
  Note that, as two $C_3$ cannot share an edge in $G$, we have that $\ct(G) \leq \frac{m(G)}{3}$.  This implies that the minimum degree of $G$ is at most $2$. Indeed, if each vertex of $G$ had degree at least $3$, then $m(G) \geq \frac{3}{2} n(G)$, which together with the relations  $\ct(G) \leq \frac{m(G)}{3}$ and $n(G)-m(G)+\ct(G) = \cc(G)$ would imply that $\cc(G) \leq 0$, a contradiction. Let $v \in V(G)$ be a vertex with minimum degree. We distinguish two cases according to the degree of $v$.

  If $v$ has degree $0$ or $1$, then the graph $G\gm \{v\}$ satisfies the $C_4$-condition as well, implying that $C_4 \not \pretp G\gm \{v\}$.
  As $v$ has degree at most one, it cannot be inside a cycle, hence $C_4 \not \pretp G$.

  Assume that $v$ has degree two and participates in a $C_3$.
  As $G$ does not contain a \ourdiamond as a subgraph, $C_4 \pretp G$ if and only if $C_4 \pretp G \gm \{v\}$.
  Moreover $n(G\gm \{v\}) = n(G) -1$, $m(G\gm \{v\}) = m(G) -2$, $\ct(G\gm \{v\}) = \ct(G) -1$, and $\cc(G \gm \{v\}) = \cc(G)$.
  This implies that $G \gm \{v\}$ satisfies the $C_4$-condition, hence $C_4 \not \pretp G \gm \{v\}$, and therefore $C_4 \not \pretp G$.

  Finally, assume that $v$ has degree two and does not belong to any $C_3$.
  Using the induction hypothesis and Lemma~\ref{brawling}, we have that
  $n(G\gm \{v\}) - m(G\gm \{v\}) + \ct(G\gm \{v\}) \leq \cc(G\gm \{v\})$.
  As $n(G\gm \{v\}) = n(G) -1$,  $m(G\gm \{v\}) = m(G) -2$, $\ct(G\gm \{v\}) = \ct(G)$, $v$ has degree two in $G$, and $G$ satisfies the $C_4$-condition, we obtain that $\cc(G\gm \{v\}) = \cc(G)-1$.
  This implies that $G \gm \{v\}$ satisfies the $C_4$-condition, and thus $C_4 \not \pretp G \gm \{v\}$.
  Since  $v$ disconnects one of the connected components of $G$ it cannot participate in a cycle of $G$, hence $C_4 \not \pretp G$.
\end{proof}

\begin{lemma}
  \label{fatherly}
  If $G$ is a non-empty graph such that $C_4 \not \pretp G$, then
  $ m(G) \leq \frac{3}{2} (n(G) -1)$.
\end{lemma}
\begin{proof}
  As $C_4 \not \pretp G$,  by Lemma~\ref{lifeboat} $G$ satisfies the $C_4$-condition.
  It follows that $\ct(G) \leq \frac{1}{3}m(G)$.
  Moreover, as $G$ is non-empty, we have that $1 \leq \cc(G)$.
  The lemma follows by using these inequalities in the equality $n(G) - m(G) + \ct(G) = \cc(G)$.
\end{proof}

We now have all the tools needed to describe our algorithm. Recall that the basic ingredients of the rank-based approach of Bodlaender et al.\cite{BodlaenderCKN15} were given in Section~\ref{banality}.
Let $G$ be a graph and $k$ be an integer.
The algorithm we describe solves the decision version of {\sc $\{C_4\}$-TM-Deletion}.
This algorithm is based on the one given in~\cite[Section 3.5]{BodlaenderCKN15}
for {\sc Feedback Vertex Set}.

We define a new graph $G_0 = (V(G) \cup \{v_0\},E(G)\cup E_0)$, where $v_0$ is a new vertex 
and $E_0 = \{\{v_0,v\} \mid v \in V(G)\}$. The role of $v_0$ is to artificially guarantee the connectivity of the solution graph, so that  the machinery of  Bodlaender et al.~\cite{BodlaenderCKN15} can be applied.
In the following, for each subgraph $H$ of $G_0$, for each $Z \subseteq V(H)$, and for each $Z_0 \subseteq E_0\cap E(H[Z])$, we denote by $H\angle{Z,Z_0}$ the graph $\big(Z,Z_0 \cup E\big(H[Z \sm \{v_0\}] 
\big)\big)$.

Given a nice tree decomposition of $G$ of width $w$, we define a nice tree decomposition  $((T,\Xcal), r, \Gcal)$ of $G_0$ of width $w+1$
such that the only empty bags are the root and the leaves and for each $t \in T$, if $X_t \not = \es$ then $v_0 \in X_t$.
Note that this can be done in linear time.
For each bag $t$,  each integers $i$, $j$, and $\ell$,
each function $\sbf: X_t \rightarrow \{0,1\}$,
each function $\sbf_0: \{v_0\} \times \sbf^{-1}(1) \rightarrow  \{0,1\}$,
each function $\rbf: E(G_t\angle{\sbf^{-1}(1),\sbf_0^{-1}(1)}) \to \{0,1\}$,
and for each partition $p \in \Pi(\sbf^{-1}(1))$, if  $C_4 \not \pretp G_t\angle{\sbf^{-1}(1),\sbf_0^{-1}(1)}$,
we define:

\begin{eqnarray*}
  \Ecal_t(p,\sbf,\sbf_0, \rbf,i,j,\ell) & = &
                                              \{(Z,Z_0) \mid (Z,Z_0) \in 2^{V_t} \times 2^{E_0 \cap E(G_t)}\\
                                        &&~~~ |Z|  = i,\ |E(G_t\angle{Z,Z_0})| = j,\  \ct(G_t\angle{Z,Z_0}) = \ell,\\
                                        &&~~~ G_t\angle{Z,Z_0} \mbox{ does not contain the \ourdiamond as a subgraph,}\\
                                        &&~~~ Z \cap X_t = \sbf^{-1}(1),\ Z_0 \cap (X_t \times X_t) = \sbf_0^{-1}(1),\ \\
                                        &&~~~~~~~~~~~~~~~~~~~~~~~~~~~~~~~~~~~~~~~~~~~~~~~~~~~~ v_0 \in X_t \Rightarrow \sbf(v_0) = 1, \\
                                        &&~~~ \forall u \in Z\sm X_t : \mbox{ either $t$ is the root or } \\
                                        &&~~~~~~~~~~~~~\exists u' \in \sbf^{-1}(1): \mbox{$u$ and $u'$ are connected in $G_t\angle{Z,Z_0}$,}\\
                                        &&~~~ \forall v_1,v_2 \in \sbf^{-1}(1):p \sqsubseteq V_t[\{v_1,v_2\}] \Leftrightarrow  \mbox{$v_1$ and $v_2$ are con-} \\
                                        &&~~~~~~~~~~~~~~~~~~~~~~~~~~~~~~~~~~~~~~~~~~~~~~~~~~~~~\mbox{nected in $G_t\angle{Z,Z_0}$},\\
                                        &&~~~ \forall e \in E(G_t\angle{Z,Z_0}) \cap {\sbf^{-1}(1) \choose 2}: \rbf(e) = 1 \Leftrightarrow \mbox{ $e$ is an edge}\\
                                        &&~~~~~~~~~~~~~~~~~~~~~~~~~~~~~~~~~~~~~~~~~~~~~~~~~~~\mbox{ of   a $C_3$ in $G_t\angle{Z,Z_0}$}\}\\
  \\
  \Acal_t(\sbf,\sbf_0, \rbf,i,j,\ell) & = &
                                            \{p \mid p \in \Pi(\sbf^{-1}(1)),\ \Ecal_t(p,\sbf,\sbf_0, \rbf,i,j,\ell) \not = \es\}.\\
\end{eqnarray*}
Otherwise, i.e., if $C_4 \pretp G_t\angle{\sbf^{-1}(1),\sbf_0^{-1}(1)}$, we define
\begin{eqnarray*}
  \Acal_t(\sbf,\sbf_0, \rbf,i,j,\ell) & = & \es.
\end{eqnarray*}
Note that we do not need to keep track of partial solutions if $C_4 \pretp G_t\angle{\sbf^{-1}(1),\sbf_0^{-1}(1)}$, as we already know they will not lead to a global solution.
Moreover, if $C_4 \not \pretp G_t\angle{\sbf^{-1}(1),\sbf_0^{-1}(1)}$, then by Lemma~\ref{fatherly},  $ m(G_t\angle{\sbf^{-1}(1),\sbf_0^{-1}(1)}) \leq \frac{3}{2} (n(G_t\angle{\sbf^{-1}(1),\sbf_0^{-1}(1)}) -1)$.

Using the definition of $\Acal_r$, Lemma~\ref{lifeboat}, and Lemma~\ref{fatherly}, we have that ${\bf tm}_{\{C_4\}}(G) \leq k$
if and only if for some $i \geq |V(G) \cup \{v_0\}| -k$ and some $j \leq \frac{2}{3}(i-1)$, we have $\Acal_r(\ef,\ef,\ef,i,j,1+j-i) \not = \es$.
For each $t \in V(T)$, we assume that we have already computed $\Acal_{t'}$ for each children $t'$ of $t$, and we proceed to the computation of $\Acal_t$.
As usual, we distinguish
several cases depending on the type of node $t$.

\begin{description}
\item[Leaf.] By definition of $\Acal_t$ we have $\Acal_t(\ef,\ef,\ef,0,0,0) = \{\es\}$.
\item[Introduce vertex.]
  Let $v$ be the insertion vertex of $X_t$, let $t'$ be the child of $t$, let $\sbf$, $\sbf_0$, and $\rbf$ the functions defined as before, let $H = G_t\angle{\sbf^{-1}(1),\sbf_0^{-1}(1)}$, and let $d_3$ be the number of $C_3$'s of $H$ that contain the vertex $v$.
  \begin{itemize}
  \item If $C_4 \pretp H$ or if $v = v_0$ and $\sbf(v_0) = 0$, then by definition of $\Acal_t$ we have that $\Acal_t(\sbf,\sbf_0, \rbf,i,j,\ell) = \es$.
  \item Otherwise, if $\sbf(v) = 0$, then, by definition of $\Acal_t$, it holds that $\Acal_t(\sbf,\sbf_0, \rbf,i,j,\ell) = \Acal_{t'}(\sbf|_{X_{t'}},\sbf_{0}|_{E_{t'}}, \rbf|_{E_{t'}},i,j,\ell)$.
  \item Otherwise, if $v = v_0$, then by construction of the nice tree decomposition,
    we know that $t'$ is a leaf of $T$ and so
    $\sbf = \{(v_0,1)\}$, $\sbf_0 = \rbf = \es$, $j = \ell = i-1 = 0$ and $\Acal_t(\sbf,\sbf_0, \rbf,i,j,\ell) = \ins(\{v_0\}, \Acal_{t'}(\ef,\ef,\ef,0,0,0)$).
  \item
    Otherwise, we know that $v \not = v_0$, $\sbf(v) = 1$, $v_0 \in N_{G[\sbf^{-1}(1)]}(v)$, and $C_4\not \pretp H$.
    As $\sbf(v) = 1$, we have to insert $v$ and we have to make sure that all vertices of $N_H[v]\sm \{v_0\}$ are in the same connected component of $H$.
    The only remaining choice is either we insert the edge $\{v,v_0\}$ or not.
    This is handle by the value of $\sbf_{0}(\{v_0,v\})$.
    So, by adding $v$, we add one vertex, $|N_H(v)|$ edges, and $d_3$  $C_3$'s.
    We also have to take care not to introduce a \ourdiamond.
    For this, the function $\rbf$ should be such that, for every edge $e$ contained in a $C_3$'s of $H$ that contains the vertex $v$, $\rbf(e) = 1$.
    We define $\rbf': E(H[X_{t'}]) \to \{0,1\}$ such that  for every edge $e\in E(H[X_{t'}])$ contained in a $C_3$'s of $H$ that contains the vertex $v$, $\rbf'(e) = 0$, and for each other edge $e$ of $H[X_{t'}]$, $\rbf'(e) = \rbf(e)$.
    Therefore, we have that 
    \begin{eqnarray*}
      \Acal_t(\sbf,\sbf_0, \rbf,i,j,\ell) & = & \\
                                          & &\!\!\!\!\!\!\!\!\!\!\!\!\!\!\!\!\!\!\!\!\!\!\!\!\!\!\!\!\!\!\!\!\!\!\!\!\!\!\!\!\!\!\!\!\!\!\!\!\!\!\!\!\!\!\! \glue(N_{H}[v],\ins(\{v\}, \Acal_{t'}(\sbf|_{X_{t'}},\sbf_{0}|_{E_{t'}}, \rbf',i-1,j - |N_H(v)|,\ell -d_3))).\\
    \end{eqnarray*}
  \end{itemize}
\item[Forget vertex.]
  Let $v$ be the forget vertex of $X_t$, let $t'$ be the child of $t$, and let $\sbf$, $\sbf_0$, and $\rbf$ the functions defined as before.
  For each function, we have a choice on how it can be extended in $t'$, and we potentially need to consider every possible such extension.
  Note the number of vertices, edges, or $C_3$'s is not affected.
  We obtain that
  \begin{eqnarray*}
    \Acal_t(\sbf,\sbf_0, \rbf,i,j,\ell) &=&
                                            A_{t'}(\sbf\cup \{(v,0)\},\sbf_0, \rbf,i,j,\ell)\\
                                        &&\!\!\!\!\!\!\!\!\!\!\!\!
                                           \underset{\rbf': E(G_t\angle{\sbf'^{-1}(1),\sbf_0'^{-1}(1)})\to \{0,1\} ,\ \rbf'|_{X_t} = \rbf}{
                                           \underset{\sbf'_{0}: \{v_0\} \times \sbf'^{-1}(1)\to \{0,1\} ,\ \sbf'_{0}|_{X_t} = \sbf_0}{
                                           \underset{\sbf': X_{t'}\to \{0,1\} ,\ \sbf'|_{X_t} = \sbf,\ \sbf'(v) = 1}{\underset{}
                                           \bigcuparrow
                                           }}}
                                           \proj(\{v\},A_{t'}(\sbf',\sbf'_0, \rbf',i,j,\ell)).
  \end{eqnarray*}
\item[Join.]
  Let $t'$ and $t''$ be the two children of $t$, let $\sbf$, $\sbf_0$, and $\rbf$ be the functions defined as before, let $H = G_t\angle{\sbf^{-1}(1),\sbf_0^{-1}(1)}$, and let $S\subseteq E(H)$ be the set of edges that participate in a $C_3$ of $H$.

  We join every compatible entries
  $A_{t'}(\sbf',\sbf'_0, \rbf',i',j',\ell')$ and $A_{t''}(\sbf'',\sbf''_0, \rbf'',i'',j'',\ell'')$.
  For two such entries being compatible, we need $\sbf' = \sbf'' = \sbf$ and $\sbf_0' = \sbf''_0 = \sbf_0$.
  Moreover, we do not want the solution graph to contain a \ourdiamond as a subgraph, and for this we need $\rbf'^{-1}(1) \cap \rbf''^{-1}(1) = S$.
  Indeed,
  either $H$ contains the \ourdiamond as a subgraph, and then $A_{t'}(\sbf',\sbf'_0, \rbf',i',j',\ell') = A_{t''}(\sbf'',\sbf''_0, \rbf'',i'',j'',\ell'') = \{\es\}$, or
  the \ourdiamond is created by joining two $C_3$'s, one from $t'$ and the other one from $t''$, sharing a common edge.
  This is possible only  if $(\rbf'^{-1}(1) \cap \rbf''^{-1}(1)) \sm S \not = \es$.
  For the counters, we have to be careful in order not to count some element twice.
  We obtain that
  \begin{eqnarray*}
    \Acal_t(\sbf,\sbf_0, \rbf,i,j,\ell) &=&\!\!\!\!\!\!\!\!\!\!
                                            \underset{\ell' + \ell'' = \ell + \ct(H)}{
                                            \underset{j'+j'' = j + |E(H)|}{
                                            \underset{i'+i''=i+|V(H)|}{
                                            \underset{\rbf'^{-1}(1) \cap \rbf''^{-1}(1) = S}{
                                            {
                                            \underset{\rbf',\rbf'': E(H)\to \{0,1\},}{
                                            \underset{}
                                            \bigcuparrow
                                            }}}}}}
                                            \join(A_{t'}(\sbf,\sbf_0, \rbf',i',j',\ell'), A_{t''}(\sbf,\sbf_0, \rbf'',i'',j'',\ell'')).
  \end{eqnarray*}

\end{description}


\begin{theorem}\label{strategy}
  {\sc $\{C_4\}$-TM-Deletion} can be solved in time $2^{\Ocal(\tw)}\cdot  n^7$.
\end{theorem}


\begin{proof}
  The algorithm works in the following way.
  For each node $t \in V(T)$ and for each entry $M$ of its table, instead of storing $\Acal_t(M)$, we store $\Acal'_t(M) = \reduce(\Acal_t(M))$ by using Theorem~\ref{happened}.
  As each of the operation we use preserves representation by Proposition~\ref{thorough}, we obtain that for each node $t \in V(T)$ and for each possible entry $M$, $\Acal'_t(M)$ represents $\Acal_t(M)$.
  In particular, we have that $\Acal'_r(M) = \reduce(\Acal_r(M))$ for each possible entry $M$.
  Using the definition of $\Acal_r$, Lemma~\ref{lifeboat}, and Lemma~\ref{fatherly}, we have that
  ${\bf tm}_{\{C_4\}}(G) \leq k$
  if and only if for some $i \geq |V(G) \cup \{v_0\}| -k$ and some $j \leq \frac{2}{3}(i-1)$, we have $\Acal'_r(\ef,\ef,\ef,i,j,1+j-i) \not = \es$.

  We now focus on the running time of the algorithm.
  The size of the intermediate sets of weighted partitions, for a leaf node and for an introduce vertex node are upper-bounded by $2^{|\sbf^{-1}(1)|}$.
  For a forget vertex node, as in the big union operation we take into consideration a unique extension of $\sbf$, at most two  possible extensions of $\sbf_0$, and at most $2^{|\sbf^{-1}(1)|}$
  possible extensions for $\rbf$,
  we obtain that
  the intermediate sets of weighted partitions have size at most $2^{|\sbf^{-1}(1)|} + 2\cdot 2^{|\sbf^{-1}(1)|} \cdot 2^{|\sbf^{-1}(1)|}\leq  2^{2|\sbf^{-1}(1)|+2}$.
  For a join node, as in the big union operation we take into consideration at most $2^{|E(H)|} $ possible functions $\rbf'$ and as many functions $\rbf''$, at most $n+|\sbf^{-1}(1)|$ choices for $i'$ and $i''$,
  at most $\frac{3}{2}(n-1) +|E(H)|$ choices for $j'$ and $j''$, and    at most $\frac{1}{2}(n-1) +\frac{1}{3}|E(H)|$ choices for $\ell'$ and $\ell''$, we obtain that
  the intermediate sets of weighted partitions have size at most $2^{|E(H)|} \cdot 2^{|E(H)|} \cdot (n+|\sbf^{-1}(1)|) \cdot (\frac{3}{2}(n-1) +|E(H)|) \cdot (\frac{1}{2}(n-1) +\frac{1}{3}|E(H)|) \cdot 4^{|\sbf^{-1}(1)|} $.
  As each time we can check the condition $C_4 \not \pretp H$, by Lemma~\ref{fatherly}   $ m(H) \leq \frac{3}{2} (n(H) -1)$, so we obtain that the intermediate sets of weighted partitions have size at most
  $6 \cdot n^3 \cdot 2^{5|\sbf^{-1}(1)|}$.
  Moreover, for each node $t \in V(T)$, the function $\reduce$ will be called as many times as the number of possible entries, i.e., at most $2^{\Ocal(w)} \cdot n^3$ times.
  Thus, using Theorem~\ref{happened},  $\Acal'_t$ can be computed in time $2^{\Ocal(w)} \cdot n^6$.
  The theorem follows by taking into account the linear number of nodes in a nice tree decomposition.
\end{proof}

\section{A single-exponential algorithm for $\{\paw\}$-{\sc TM-Deletion}}
\label{nugatory}



Again, we start with a simple structural characterization of the simple graphs that exclude the \paw as a topological minor; recall the $\paw$ graph in Figure~\ref{sleepily}.

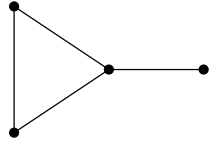
\begin{figure}[htb]
  \centering
  \scalebox{1.2}{\begin{tikzpicture}[scale=.7]
      \vertex{0,0};
      \vertex{0,2};
      \vertex{1.5,1};
      \vertex{3,1};

      \draw (1.5,1) -- (0,0) -- (0,2) -- (1.5,1) -- (3,1);
    \end{tikzpicture}}
  \caption{The \paw graph.}
  \label{sleepily}
\end{figure}

\begin{lemma}
  \label{absurdly}
  Let $G$ be a simple graph.
  $\paw \not \pretp G$ if and only if each connected component of $G$ is either a cycle or a tree.
\end{lemma}
\begin{proof}
  It is easy to see that neither a cycle nor a tree contain the \paw as a topological minor.
  Let $G$ be a graph such that  $\paw \not \pretp G$. Let us assume w.l.o.g. that $G$ is connected. If $G$ does not contains a cycle, then it is a tree. Otherwise, let $C$ be a chordless  cycle in $G$. If $G$ contains a vertex $v$ that is not in $C$ then, as $G$ is connected, there exists a path from $v$ to $C$ containing at least two vertices. This is not possible, as it would imply that $G$ contains the \paw  as a topological minor. As $C$ is chordless and $G$ is simple, we obtain that $G$ is exactly the cycle $C$, and the lemma follows.
\end{proof}

We present an algorithm that solves the  decision version of {\sc $\{\paw\}$-TM-Deletion}.
As the algorithm that we presented for {\sc $\{C_4\}$-TM-Deletion} in Section~\ref{benjamin}, this algorithm is based on the one given in~\cite[Section 3.5]{BodlaenderCKN15}
for {\sc Feedback Vertex Set}.
Let $G$ be a graph and $k$ be an integer.
The idea of the following algorithm is to partition $V(G)$ into three sets. The first one will be the solution set $S$, the second one will be a set $F$ of vertices that induces a forest, and the third one will be a set $C$ of vertices that induces a collection of cycles. If we can partition our graph into three such sets $(S,F,C)$ such that there is no edge between a vertex of $F$ and a vertex of $C$ and such that $|S| \leq k$, then, using Lemma~\ref{absurdly}, we know that   ${\bf tm}_{\{\paw\}}(G) \leq k$. On the other hand, if such a partition does not exist, we know that   ${\bf tm}_{\{\paw\}}(G) > k$.
The main idea of this algorithm is to combine classical dynamic programming techniques in order to verify that $C$ induces a collection of cycles, and the rank-based approach in order to verify that $F$ induces a forest.

As for {\sc $\{C_4\}$-TM-Deletion}, we define a new graph $G_0 = (V(G) \cup \{v_0\},E(G)\cup E_0)$, where $v_0$ is a new vertex
and $E_0 = \{\{v_0,v\} \mid v \in V(G)\}$. 
We recall that for each subgraph $H$ of $G_0$, for each $Z_1 \subseteq V(H)$, and for each $Y \subseteq E_0\cap E(H[Z_1])$, we denote by $H\angle{Z_1,Y}$ the graph $\big(Z_1,Y \cup E\big(H[Z_1 \sm \{v_0\}]
\big)\big)$.

Given a nice tree decomposition of $G$ of width $w$, we define a nice tree decomposition  $((T,\Xcal), r, \Gcal)$ of $G_0$ of width $w+1$
such that the only empty bags are the root and the leaves and for each $t \in T$, if $X_t \not = \es$ then $v_0 \in X_t$.
Note that this can be done in linear time.
For each bag $t$,  each integers $i$, $j$, and $\ell$,
each function $\sbf: X_t \rightarrow \{0,1,2_0,2_1,2_2\}$,
each function $\sbf_0: \{v_0\} \times \sbf^{-1}(1) \rightarrow \{0,1\}$,
and each partition $p \in \Pi(\sbf^{-1}(1))$,
we define:

\begin{eqnarray*}
  \Ecal_t(p,\sbf, \sbf_0,i,j,\ell) & = &
                                         \{(Z_1,Z_2,Y) \mid (Z_1,Z_2,Y) \in 2^{V_t} \times  2^{V_t} \times 2^{E_0 \cap E(G_t)},\ Z_1 \cap Z_2 = \es,\\
                                   &&~~~ |Z_1|  = i,\ |Z_2| = \ell, |E(G_t[Z_1\sm \{v_0\}]) \cup Y| = j, \\
                                   &&~~~ \forall e \in E_0 \cap E_t,\ \sbf_0(e) = 1 \Leftrightarrow e \in Y,\\
                                   && ~~~ \forall v \in Z_2 \cap X_t,\ \sbf(v) = 2_z\mbox{ with }z = {\degs{G_t[Z_2]}{v}},\\
                                   && ~~~ \forall v \in Z_2 \sm X_t,\ {\degs{G_t[Z_2]}{v} = 2},\\
                                   &&~~~ Z_1 \cap X_t = \sbf^{-1}(1),\ v_0 \in X_t \Rightarrow \sbf(v_0) = 1, \\ 
                                   &&~~~\mbox{$G_t\angle{\sbf^{-1}(1),\sbf_0^{-1}(1)}$ is a forest},\\
                                   &&~~~ \forall u \in Z_1\sm X_t : \mbox{ either $t$ is the root or } \\
                                   &&~~~~~~~~~~~~~\exists u' \in \sbf^{-1}(1): \mbox{$u$ and $u'$ are connected in $G_t\angle{Z_1,Y}$,}\\
                                   &&~~~ \forall v_1,v_2 \in \sbf^{-1}(1):p \sqsubseteq V_t[\{v_1,v_2\}] \Leftrightarrow  \mbox{$v_1$ and $v_2$ are con-} \\
                                   &&~~~~~~~~~~~~~~~~~~~~~~~~~~~~~~~~~~~~~~~~~~~~~~~~~~~~~\mbox{nected in $G_t\angle{Z_1,Y}$},\\
                                   && ~~~ \forall (u,v) \in (Z_1\sm \{v_0\}) \times Z_2,\ \{u,v\} \not \in E(G_t) \}\\
  \\
  \Acal_t(\sbf, \sbf_0,i,j,\ell) & = &
                                       \{p \mid p \in \Pi(\sbf^{-1}(1)),\ \Ecal_t(p,\sbf,\sbf_0,i,j,\ell) \not = \es\}.\\
\end{eqnarray*}


In the definition of $\Ecal_t$, the sets $Z_1$ (resp. $Z_2$) correspond to the set $F$ (resp. $C$) restricted to $G_t$. The vertex $v_0$ and the set $Y$ exist to ensure that $F$ will be connected.

By Lemma~\ref{absurdly}, we have that the given instance of {\sc $\{\paw\}$-TM-Deletion} is a \textsc{Yes}-instance if and only if for some $i$ and $\ell$, $i+\ell \geq |V(G) \cup \{v_0\}| -k$ and $\Acal_r(\ef,\ef,i,i-1,\ell) \not = \es$.
For each $t \in V(T)$, we assume that we have already computed $\Acal_{t'}$ for every children $t'$ of $t$, and we proceed to the computation of $\Acal_t$.
As usual, we distinguish
several cases depending on the type of node $t$.

\begin{description}
\item[Leaf.] By definition of $\Acal_t$, we have $\Acal_t(\ef,\ef,0,0,0) = \{\es\}$.
\item[Introduce vertex.]
  Let $v$ be the insertion vertex of $X_t$, let $t'$ be the child of $t$, let $\sbf: X_t \rightarrow \{0,1,2_0,2_1,2_2\}$,  $\sbf_0: \{v_0\} \times \sbf^{-1}(1) \rightarrow \{0,1\}$, and let $H = G_t\angle{\sbf^{-1}(1),\sbf_0^{-1}(1)}$.
  \begin{itemize}
  \item If  $v = v_0$ and $\sbf(v_0) \in \{0,2_0,2_1,2_2\}$ or if $H$ contains a cycle, then by definition of $\Acal_t$ we have that $\Acal_t(\sbf, \sbf_0,i,j,\ell) = \es$.
  \item Otherwise, if $v = v_0$, then by construction of the nice tree decomposition,
    we know that $t'$ is a leaf of $T$ and so
    $\sbf = \{(v_0,1)\}$, $j = \ell = i-1 = 0$ and $\Acal_t(\sbf, \sbf_0,i,j,\ell) = \ins(\{v_0\}, \Acal_{t'}(\ef,\ef,0,0,0)$).
  \item Otherwise, if $\sbf(v) = 0$, then, by definition of $\Acal_t$, it holds that $\Acal_t(\sbf, \sbf_0,i,j,\ell) = \Acal_{t'}(\sbf|_{X_{t'}},\sbf|_{E_{t'}},i,j,\ell)$.
  \item Otherwise, if $\sbf(v) = 2_z$, $z \in \{0,1,2\}$, then let $Z_2'=N_{G_t[X_t]}(v)\sm \sbf^{-1}(0)$.
    If  $Z_2' \not \subseteq \sbf^{-1}(\{2_1,2_2\})$ or $|Z'_2| \not = z$ then   $\Acal_t(\sbf, \sbf_0,i,j,\ell) = \es$.
    Otherwise $Z_2' \subseteq \sbf^{-1}(\{2_1,2_2\})$ and $|Z'_2| = z$, and with
    $\sbf': X_{t'} \to \{0,1,2_0,2_1,2_2\}$ defined such that $\forall v' \in X_{t'} \sm Z'_2,\ \sbf'(v') = \sbf(v')$ and for each $v' \in Z_2'$ such that $\sbf(v') = 2_{z'}$, $z' \in \{1,2\}$, $\sbf'(v') = 2_{z'-1}$. It holds that
    $\Acal_t(\sbf, \sbf_0,i,j,\ell) = \Acal_{t'}(\sbf',\sbf_0 ,i,j,\ell-1)$.
  \item
    Otherwise, we know that $v \not = v_0$, $\sbf(v) = 1$, and $v_0 \in N_{G[\sbf^{-1}(1)]}(v)$.
    First, if $N_{G_t[X_t]}(v) \sm \sbf^{-1}(0) \not \subseteq \sbf^{-1}(1)$, then $\Acal_t(\sbf,\sbf_0,i,j,\ell) = \es$. Indeed, this implies that the cycle part and the forest part are connected.
    As $\sbf(v) = 1$, we have to insert $v$ in the forest part and we have to make sure that all vertices of $N_H[v]$ are in the same connected component of $H$.
    The only remaining choice is to insert the edge $\{v,v_0\}$ or not.
    Again, this is handled by the function $\sbf_0$.
    By adding $v$, we add one vertex and $|N_H(v)|$ edges in the forest part.
    Therefore, we have that 
    \begin{eqnarray*}
      \Acal_t(\sbf,\sbf_0, i,j,\ell) & = & \\
                                          & &\!\!\!\!\!\!\!\!\!\!\!\!\!\!\!\!\!\!\!\!\!\!\!\!\!\!\!\!\!\!\!\!\!\!\!\!\!\!\!\!\!\!\!\!\!\!\!\!\!\!\!\!\!\!\! \glue(N_{H}[v],\ins(\{v\}, \Acal_{t'}(\sbf|_{X_{t'}},\sbf_0|_{E_{t'}},i-1,j - |N_H(v)|,\ell))).\\
    \end{eqnarray*}
  \end{itemize}
\item[Forget vertex.]
  Let $v$ be the forget vertex of $X_t$, let $t'$ be the child of $t$, and let  $\sbf: X_t \rightarrow \{0,1,2_0,2_1,2_2\}$.
  As a vertex from the collection of cycles can be removed only if it has exactly two neighbors, we obtain that
  \begin{eqnarray*}
    \Acal_t(\sbf,\sbf_0,i,j,\ell) &=&
                               A_{t'}(\sbf\cup \{(v,0)\},\sbf_0,i,j,\ell)\\
                           &&\cuparrow \proj(\{v\},A_{t'}(\sbf \cup \{(v,1)\},\sbf_0 \cup \{(\{v_0,v\},0)\},i,j,\ell))\\
                           &&\cuparrow \proj(\{v\},A_{t'}(\sbf \cup \{(v,1)\},\sbf_0 \cup \{(\{v_0,v\},1)\},i,j,\ell))\\
                           && \cuparrow A_{t'}(\sbf\cup \{(v,2_2)\},\sbf_0,i,j,\ell).\\
  \end{eqnarray*}
\item[Join.]
  Let $t'$ and $t''$ be the two children of $t$, let $\sbf: X_t \rightarrow \{0,1,2_0,2_1,2_2\}$,  $\sbf_0: \{v_0\} \times \sbf^{-1}(1) \rightarrow \{0,1\}$, and let $H = G_t\angle{\sbf^{-1}(1),\sbf_0^{-1}(1)}$.
  Given three functions $\sbf^*,\sbf', \sbf'': X_t \rightarrow \{0,1,2_0,2_1,2_2\}$, we say that $\sbf^* = \sbf' \oplus \sbf''$ if for each $v \in \sbf^{-1}(\{0,1\})$, $\sbf^*(v) = \sbf'(v) = \sbf''(v)$, and for each $v \in X_t$ such that $\sbf^*(v) = 2_z$, $z \in \{0,1,2\}$, there exist $z',z'' \in \{0,1,2\}$ such that $\sbf'(v) = 2_{z'}$, $\sbf''(v)=2_{z''}$, and $z = z' + z'' - {\degs{G_t[X_t\sm \sbf^{-1}(0)]}{v}} $.

  We join every compatible entries
  $A_{t'}(\sbf',\sbf'_0,i',j',\ell')$ and $A_{t''}(\sbf'',\sbf''_0,i'',j'',\ell'')$.
  For two such entries being compatible, we need $\sbf' \oplus \sbf''$ to be defined and $\sbf_0' = \sbf_0''$.
  We obtain that
  \begin{eqnarray*}
    \Acal_t(\sbf,\sbf_0,i,j,\ell) &=&\!\!\!\!\!\!\!\!\!\!
                                      \underset{\ell' + \ell'' = \ell + |\sbf^{-1}(\{2_0,2_1,2_2\})|}{
                                      \underset{j'+j'' = j + |E(H)|}{
                                      \underset{i'+i''=i+|V(H)|}{
                                      \underset{\sbf = \sbf_1 \oplus \sbf_2}{
                                      {
                                      \underset{\sbf', \sbf'': X_t \rightarrow \{0,1,2_0,2_1,2_2\},}{
                                      \underset{}
                                      \bigcuparrow
                                      }}}}}}\!\!\!\!\!\!\!
                                      \join(A_{t'}(\sbf',\sbf_0,i',j',\ell'), A_{t''}(\sbf'',\sbf_0, i'',j'',\ell'')).
  \end{eqnarray*}

\end{description}


\begin{theorem}
  {\sc $\{\paw\}$-TM-Deletion} can be solved in time $2^{\Ocal(\tw)}\cdot  n^7$.
\end{theorem}


\begin{proof}
  The algorithm works in the following way.
  For each node $t \in V(T)$ and for each entry $M$ of its table, instead of storing $\Acal_t(M)$, we store $\Acal'_t(M) = \reduce(\Acal_t(M))$ by using Theorem~\ref{happened}.
  As each of the operations we use preserves representation by Proposition~\ref{thorough}, we obtain that for each node $t \in V(T)$ and for each possible entry $M$, $\Acal'_t(M)$ represents $\Acal_t(M)$.
  In particular, we have that $\Acal'_r(M) = \reduce(\Acal_r(M))$ for each possible entry $M$.
  Using the definition of $\Acal_r$ and Lemma~\ref{absurdly}, we have that
  ${\bf tm}_{\{\paw\}}(G) \leq k$ if and only if for some $i$ and $\ell$, $i+\ell \geq |V(G) \cup \{v_0\}| -k$ and $\Acal'_r(\ef,\ef,i,i-1,\ell) \not = \es$.


  We now focus on the running time of the algorithm.
  The size of the intermediate sets of weighted partitions for a leaf node and for an introduce vertex node, are upper-bounded by $2^{|\sbf^{-1}(1)|}$.
  For a forget vertex node, we take the union of four sets of size $2^{|\sbf^{-1}(1)|}$, so
  the intermediate sets of weighted partitions have size at most $4 \cdot 2^{|\sbf^{-1}(1)|}$.
  For a join node, as in the big union operation we take into consideration at most $5^{|X_t|} $ possible functions $\sbf'$, as many functions $\sbf''$, at most $n+|\sbf^{-1}(1)|$ choices for $i'$ and $i''$,
  at most $n+|\sbf^{-1}(1)|$ choices for $j'$ and $j''$ (as  $H$ is always a forest), and at most $n+|\sbf^{-1}(\{2_0,2_1,2_2\}|$ choices for $\ell'$ and $\ell''$, we obtain that the intermediate sets of weighted partitions have size at most
  $25^{|X_t|}\cdot (n+ |\sbf^{-1}(1)|)^2 \cdot (n+|\sbf^{-1}(\{2_0,2_1,2_2\}|) \cdot 4^{|\sbf^{-1}(1)|}$.
  We obtain that the intermediate sets of weighted partitions have size at most
  $ (n+|X_t|)^3 \cdot 100^{|X_t|}$.
  Moreover, for each node $t \in V(T)$, the function $\reduce$ will be called as many times as the number of possible entries, i.e., at most $2^{\Ocal(w)} \cdot n^3$ times.
  Thus, using Theorem~\ref{happened},  $\Acal'_t$ can be computed in time $2^{\Ocal(w)} \cdot n^6$.
  The theorem follows by taking into account the linear number of nodes in a nice tree decomposition.
\end{proof}

\section{A single-exponential algorithm for {\sc $\{\chair\}$-TM-Deletion}}
\label{provider}

As in the previous cases that we solved in single-exponential time, we start with a structural characterization of the graphs that exclude the $\chair$ as a topological minor (hence, as a minor as well).

\begin{lemma}
  \label{fortress}
  Let $G$ be a graph.
  $\chair \not \pretp G$ if and only if
  every connected component of $G$ of size at least five is a path, a cycle, or a star.
\end{lemma}
\begin{proof}
  It is straightforward to check that a path, a cycle, a star, or a graph of size at most $4$ do not contain the $\chair$ as a topological minor.
  Conversely, let $G$ be a connected graph of size at least five that excludes the chair as a topological minor.
  Let $P$ be a longest path of $G$. We denote by $p_1$ and $p_2$ the endpoints of this path $P$.
  It is easy to see that if $|V(P)| \leq 3$, then $G$ has to be a star.
  Assume now that $|V(P)| \geq 4$.
  Then we have that $V(P) = V(G)$.
  Indeed, assume that $V(P) \not = V(G)$ and let $v \in V(G) \setminus V(P)$ be a vertex adjacent with a vertex of $V(P)$ in $G$.
  This vertex should exist by the connectivity of $G$.
  By maximality of $P$, $v$ cannot be adjacent to $p_1$ or $p_2$ and as $\chair \not\pretp G$, $v$ cannot be adjacent to an internal vertex of the path $P$.
  Thus, $|V(P)| = |V(G)| \geq 5$.
  Moreover, as $\chair \not\pretp G$, neither $p_1$ nor $p_2$ can be adjacent to an internal vertex of the path $P$, and so,
  $E(P) \subseteq E(G) \subseteq E(P) \cup \{p_1,p_2\}$.
  Thus, $G$ is either a path or a cycle.
\end{proof}

With Lemma~\ref{fortress} at hand, an algorithm for {\sc $\{\chair\}$-TM-Deletion} running in time $2^{\Ocal(\tw)}\cdot n^{\Ocal(1)}$ can be obtained using standard dynamic programming techniques.
For this, as we did for {\sc $\{P_3\}$-TM-Deletion}, {\sc $\{P_4\}$-TM-Deletion}, and {\sc $\{K_{1,s}\}$-TM-Deletion}, $s \in \Nbb$, we  label the vertices on each bag.
Each label carries two types of information. On the one hand, it indicates whether a vertex is in a collection of paths or cycles, a collection of stars, a clique of size $4$, a \paw, or a \ourdiamond; note that these are all the possible graphs on at most 4 vertices (see Figure~\ref{shifting}).  On the other hand, it also indicates in which ``state'' a vertex is with regard to the already computed graph, which will become clear below.

As the number of distinct labels needed for
the algorithm for {\sc $\{\chair\}$-TM-Deletion} is quite large and the algorithm itself is not complicated,
we will avoid the formal description of it.
Namely, as we will discuss, we need $17$ distinct labels on the vertices and three distinct labels on the edges.
We only describe how to label the vertices for each type of component.
Note that, for instance, a $P_3$ is both a path and a star.
We do not consider this as an issue, and we will just have two types of components that can become $P_3$'s.

Let us proceed to the description of the labels.
First, we use one label to indicate whether a vertex belongs in the solution or not.
For the collection of paths or cycles, we use three labels $0$, $1$, and $2$, corresponding to the current degree of this vertex.
For a collection of stars, we use three labels, $c$, $0$, $1$ where $c$ labels a vertex that is a center of a star,
$0$ labels a vertex that is leaf of a star that is not connected to a center yet, and
$1$ labels a leaf of a star that is already connected to a center.
For the clique of size $4$, we only need two labels $0$ and $1$, where $0$ means that the vertex should be in a $K_4$ but we do not know which one yet, and $1$ means that we already found in which $K_4$ the vertex is.
The crucial argument for this is the fact that if four vertices induce a $K_4$, then by the properties of a tree decomposition there exists a bag that contains the four vertices.
For the \paw, we use six labels $c_0$, $c_1$, $d_f$, $d_w$, $\ell_0$, and $\ell_1$.
We call \emph{leaf} the only vertex of the \paw of degree $1$.
The label $\ell_0$ (resp. $\ell_1$) corresponds to a leaf of a \paw that has degree $0$ (resp. $1$)  in the currently processed graph.
The label $c_0$ (resp. $c_1$) corresponds to a vertex of the cycle of a \paw that is not (resp. is) connected to a leaf and for which we do not know yet the three vertices of the cycle.
The label $d_f$ corresponds to a vertex of the cycle of a \paw for which we already know the three vertices of the cycle and that is \emph{full}, i.e., it cannot be connected to a leaf anymore.
The label $d_w$ corresponds to a vertex of the cycle of a \paw for which we know the three vertices of the cycle and that is waiting for a leaf to be connected to.

Dealing with the \ourdiamond is a bit more complicated.
We see the \ourdiamond as two $C_3$ glued by an edge called \emph{chord}.
The crucial argument is that for each $C_3$, there exists a bag that contains the three vertices of the cycle and we ``only'' need to remember which edge of the cycle is the chord of the \ourdiamond.
We use two labels for the vertices $a_0$ and $a_1$, and, this time, we also use three labels on the edges $b_0$, $b_1$, and $b_2$.
Note that in a \ourdiamond, the number of edges is linear in the number of vertices and so, we can afford to label the edges.
Namely, $a_0$ is used for vertices that can still  be in a $C_3$ and $a_1$ is used for vertices that cannot be in a new $C_3$ anymore.
The label $b_0$ is used for edges that are not in a $C_3$ yet, the label $b_1$ is used for chords for which we only found one of the two $C_3$'s, and the label $b_2$ is used for edges that cannot be used anymore.
Note that the endpoints of a chord belong to two $C_3$'s, so when detecting the first one, these endpoints will remain labeled $a_0$ but the chord will be now labeled $b_1$.

Using all these labels, and updating them in a bottom-up fashion in a tree decomposition in a standard way, we obtain the following theorem.

\begin{theorem}\label{magician}
  If a nice tree decomposition of $G$ of width $w$ is given, {\sc $\{\chair\}$-Deletion} can be solved in time $2^{\Ocal(w)} \cdot n$.
\end{theorem}

\section{A single-exponential algorithm for {\sc $\{\banner\}$-TM-Deletion}}
\label{sec:banner}


Similarly as before, we start with a structural characterization of the graphs that exclude the $\banner$ as a (topological) minor.

\begin{lemma}\label{voltaire}
  Let $G$ be a graph with at least five  vertices. $\banner \not \pretp G$ if and only if $G$ is a cycle or $C_4 \not \pretp G$.
\end{lemma}

\begin{proof}
  Let $G$ be a connected graph with $|V(G)| \geq 5$.
  It is straightforward to see that if $G$ is a cycle or if $C_4 \not \pretp G$, then $\banner \not \pretp G$.
  Conversely, assume that $G$ excludes the banner as a minor and contains a cycle $C$ of size at least four.
  As we can assume that $G$ is connected and excludes the banner as a minor, we have that $V(G) = V(C)$.
  Indeed, if there exists a vertex $v \in V(G) \setminus V(C)$, then there exists a vertex $v' \in V(G') \setminus V(C)$ that is a neighbor of a vertex of $V(C)$, and so the graph $G[V(C) \cup \{v'\}]$ contains the banner as a topological minor.
  By assumption, we have that $|V(C)| = |V(G)| \geq 5$.
  Assume now for contradiction that $G$ is {\sl not} a cycle, that is, that $E(G) \setminus E(C)$ contains an edge $e$.
  Then, as $|V(C)| \geq 5$, there exists a cycle $C'$ containing $e$, such that $4 \leq |V(C')| < |V(C)|$, and so $G$ contains the banner as a minor, a contradiction.
  The lemma follows.
\end{proof}


The idea of the algorithm is, as we did for {\sc $\{\paw\}$-TM-Deletion} (see Section~\ref{nugatory}), to use the structural properties given by Lemma~\ref{voltaire} and to combine the rank-based approach and the standard dynamic programming techniques.
As for {\sc $\{\chair\}$-TM-Deletion} (see Section~\ref{provider}), we will need, in particular, to label vertices that appear in components of size at most $4$ and in components that are cycles.
As  for {\sc $\{C_4\}$-TM-Deletion} (see Section~\ref{benjamin}), we also use  one label for the solution and another label for the components that exclude $C_4$ as a minor.
This means that, for this algorithm, we need $15$ labels for the vertices and three labels for the edges.
Because of this, again we do not provide the full formal description of the algorithm, and instead we provide a high-level description of how it works, reusing what we presented previously.

Namely, we explain how, starting from the algorithm for {\sc $\{C_4\}$-TM-Deletion} given in Section~\ref{benjamin}, we obtain the desired algorithm for {\sc $\{\banner\}$-TM-Deletion}.
As we did for {\sc $\{\paw\}$-TM-Deletion}, we modify the function $\sbf: X_t \to \{0,1\}$ to a function $\sbf: X_t \to \{0,1\} \cup L$, where $L$ corresponds to the set of labels needed for detecting if a component is of size at most $4$ or a cycle.
We also add a function $\rbf_{d}$ for the labeling of the edges that appear in a \ourdiamond.
Then, as we did for {\sc $\{\paw\}$-TM-Deletion}, when doing the dynamic programming operations, we use the rank-based approach for the elements of $\sbf^{-1}(1)$ and standard dynamic programming operations for the elements of $\sbf^{-1}(L)$.
Doing this, we obtain the following theorem.

\begin{theorem}
  If a nice tree decomposition of $G$ of width $w$ is given, {\sc $\{\banner\}$-Deletion} can be solved in time $2^{\Ocal(w)} \cdot n^{\Ocal(1)}$.
\end{theorem}



\vspace{-.45cm}

\section{Conclusions and further research}
\label{upheaval}

\vspace{-.15cm}

We presented single-exponential algorithms for  \textsc{$\{H\}$-TM-Deletion}  taking as  parameter the treewidth of the input graph, when $H \in \{P_3, P_4, K_{1,i}, C_4, \paw, \chair, \banner\}$. These algorithms, combined with the results of~\cite{monster1,monster3,BasteST20-SODA,SODA-arXiv}, settle completely the complexity of \textsc{$\{H\}$-M-Deletion} when $H$ is connected; see Figure~\ref{shifting} for an illustration.

Concerning the topological minor version, in order to establish a dichotomy for \textsc{$\{H\}$-TM-Deletion} when $H$ is planar and connected, it remains to obtain algorithms in time $\Ostar(2^{\Ocal(\tw  \cdot \log \tw)})$ for the graphs $H$ with maximum degree at least four, like the \gem or the \dart (see Figure~\ref{shifting}), as for those graphs the algorithm in time $\Ostar(2^{\Ocal(\tw  \cdot \log \tw)})$ given in~\cite{monster1} cannot be applied.

Our algorithms for \textsc{$\{C_i\}$-Deletion} given here and in \cite{monster1} may also be used to devise approximation algorithms for hitting or packing long cycles in a graph (in the spirit of~\cite{Chatzidimitriou15} for other patterns),  by using the fact that cycles of length at least $i$ satisfy the Erd{\H{o}}s-P{\'{o}}sa property~\cite{FioriniH14}. We did not focus on optimizing the degree of the polynomials or the constants involved in our algorithms. Concerning the latter, one could use the framework by Lokshtanov et al.~\cite{LokshtanovMS11} to prove lower bounds based on the {\sl Strong} \ETH.

%

\vspace{-.25cm}
%
%

\bibliographystyle{abbrv}
\bibliography{Biblio-Fdeletion}

\begin{thebibliography}{10}

\bibitem{P3-cover-improved}
Z.~Bai, J.~Tu, and Y.~Shi.
\newblock {An improved algorithm for the vertex cover $P_3$ problem on graphs
  of bounded treewidth}.
\newblock {\em CoRR}, abs/1603.09448, 2016.

\bibitem{BasteST17}
J.~Baste, I.~Sau, and D.~M. Thilikos.
\newblock Optimal algorithms for hitting (topological) minors on graphs of
  bounded treewidth.
\newblock In {\em Proc. of the 12th International Symposium on Parameterized
  and Exact Computation (IPEC)}, volume~89 of {\em LIPIcs}, pages 4:1--4:12,
  2017.

\bibitem{BasteST18}
J.~Baste, I.~Sau, and D.~M. Thilikos.
\newblock A complexity dichotomy for hitting small planar minors parameterized
  by treewidth.
\newblock In {\em Proc. of the 13th International Symposium on Parameterized
  and Exact Computation (IPEC)}, volume 115 of {\em LIPIcs}, pages 2:1--2:13,
  2018.

\bibitem{SODA-arXiv}
J.~Baste, I.~Sau, and D.~M. Thilikos.
\newblock {Hitting minors on bounded treewidth graphs. IV. An optimal
  algorithm}.
\newblock {\em CoRR}, abs/1907.04442, 2019.

\bibitem{BasteST20-SODA}
J.~Baste, I.~Sau, and D.~M. Thilikos.
\newblock A complexity dichotomy for hitting connected minors on bounded
  treewidth graphs: the chair and the banner draw the boundary.
\newblock In {\em Proc. of the 2020 {ACM-SIAM} Symposium on Discrete Algorithms
  (SODA)}, pages 951--970, 2020.

\bibitem{monster1}
J.~Baste, I.~Sau, and D.~M. Thilikos.
\newblock {Hitting minors on bounded treewidth graphs. I. General upper
  bounds}.
\newblock {\em {SIAM} Journal on Discrete Mathematics}, 34(3):1623--1648, 2020.

\bibitem{monster3}
J.~Baste, I.~Sau, and D.~M. Thilikos.
\newblock {Hitting minors on bounded treewidth graphs. {III.} Lower bounds}.
\newblock {\em Journal of Computer and System Sciences}, 109:56--77, 2020.

\bibitem{BeKa18}
B.~Bergougnoux and M.~M. Kant{\'{e}}.
\newblock Rank based approach on graphs with structured neighborhood.
\newblock {\em CoRR}, abs/1805.11275, 2018.

\bibitem{BodlaenderCKN15}
H.~L. Bodlaender, M.~Cygan, S.~Kratsch, and J.~Nederlof.
\newblock Deterministic single exponential time algorithms for connectivity
  problems parameterized by treewidth.
\newblock {\em Information and Computation}, 243:86--111, 2015.

\bibitem{BodlaenderDDFLP16}
H.~L. Bodlaender, P.~G. Drange, M.~S. Dregi, F.~V. Fomin, D.~Lokshtanov, and
  M.~Pilipczuk.
\newblock {A $c^k n$ $5$-Approximation Algorithm for Treewidth}.
\newblock {\em {SIAM} Journal on Computing}, 45(2):317--378, 2016.

\bibitem{Chatzidimitriou15}
D.~Chatzidimitriou, J.~Raymond, I.~Sau, and D.~M. Thilikos.
\newblock {An $O(\log OPT)$-Approximation for Covering and Packing Minor Models
  of $\Theta_r$}.
\newblock {\em Algorithmica}, 80(4):1330--1356, 2018.

\bibitem{Courcelle90}
B.~Courcelle.
\newblock {The Monadic Second-Order Logic of Graphs. I. Recognizable Sets of
  Finite Graphs}.
\newblock {\em Information and Computation}, 85(1):12--75, 1990.

\bibitem{CyganFKLMPPS15}
M.~Cygan, F.~V. Fomin, L.~Kowalik, D.~Lokshtanov, D.~Marx, M.~Pilipczuk,
  M.~Pilipczuk, and S.~Saurabh.
\newblock {\em Parameterized Algorithms}.
\newblock Springer, 2015.

\bibitem{CyganNPPRW11}
M.~Cygan, J.~Nederlof, M.~Pilipczuk, M.~Pilipczuk, J.~M.~M. van Rooij, and
  J.~O. Wojtaszczyk.
\newblock {Solving Connectivity Problems Parameterized by Treewidth in Single
  Exponential Time}.
\newblock In {\em Proc. of the 52nd Annual {IEEE} Symposium on Foundations of
  Computer Science (FOCS)}, pages 150--159, 2011.

\bibitem{Die10}
R.~Diestel.
\newblock {\em {Graph Theory}}, volume 173.
\newblock Springer-Verlag, 4th edition, 2010.

\bibitem{DF13}
R.~G. Downey and M.~R. Fellows.
\newblock {\em Fundamentals of Parameterized Complexity}.
\newblock Texts in Computer Science. Springer, 2013.

\bibitem{FioriniH14}
S.~Fiorini and A.~Herinckx.
\newblock A tighter {Erd{\H{o}}s-P{\'{o}}sa} function for long cycles.
\newblock {\em Journal of Graph Theory}, 77(2):111--116, 2014.

\bibitem{FominLPS16}
F.~V. Fomin, D.~Lokshtanov, F.~Panolan, and S.~Saurabh.
\newblock Efficient computation of representative families with applications in
  parameterized and exact algorithms.
\newblock {\em Journal of the {ACM}}, 63(4):29:1--29:60, 2016.

\bibitem{ImpagliazzoP01}
R.~Impagliazzo, R.~Paturi, and F.~Zane.
\newblock Which problems have strongly exponential complexity?
\newblock {\em Journal of Computer and System Sciences}, 63(4):512--530, 2001.

\bibitem{Klo94}
T.~Kloks.
\newblock {\em Treewidth. Computations and Approximations}.
\newblock Springer-Verlag LNCS, 1994.

\bibitem{LokshtanovMS11}
D.~Lokshtanov, D.~Marx, and S.~Saurabh.
\newblock Lower bounds based on the exponential time hypothesis.
\newblock {\em Bulletin of the {EATCS}}, 105:41--72, 2011.

\bibitem{PaikRS94}
D.~Paik, S.~M. Reddy, and S.~Sahni.
\newblock Deleting vertices to bound path length.
\newblock {\em {IEEE} Transactions on Computers}, 43(9):1091--1096, 1994.

\bibitem{P3-cover}
J.~Tu, L.~Wu, J.~Yuan, and L.~Cui.
\newblock {On the vertex cover $P_3$ problem parameterized by treewidth}.
\newblock {\em Journal of Combinatorial Optimization}, 34(2):414--425, 2017.

\end{thebibliography}

\end{document}